\documentclass{amsart}

\usepackage{amsfonts,amssymb,amsmath,eucal,pinlabel}
\usepackage[all]{xy}

\newtheorem{theorem}{Theorem}[section]
\newtheorem*{thmintro}{Theorem}
\newtheorem*{Kast}{Kasteleyn's Theorem} 
\newtheorem{lemma}[theorem]{Lemma}
\newtheorem{proposition}[theorem]{Proposition}
\newtheorem{corollary}[theorem]{Corollary}
\theoremstyle{remark}\newtheorem{remark}[theorem]{Remark}
\theoremstyle{remark}\newtheorem*{example}{Example}

\newenvironment{romanlist}
        {\begin{enumerate}
        }
        {\end{enumerate}}

\newcounter{ticklistc}
\newenvironment{ticklist}
    	{\setcounter{ticklistc}{0}
	 \begin{list}{--}
	{\usecounter{ticklistc}}}{\end{list}}

\newcommand{\Z}{\mathbb Z}
\newcommand{\R}{\mathbb R}

\newcommand{\K}{\mathcal K}

\newcommand{\D}{\mathcal D}
\newcommand{\Q}{\mathcal Q}
\newcommand{\G}{\Gamma}
\newcommand{\SI}{\Sigma}
\newcommand{\ho}{H_1(\Sigma;\Z_2)}
\newcommand{\coho}{H^1(\Sigma;\Z_2)}
\newcommand{\Hom}{\mathrm{Hom}}
\newcommand{\e}{\varepsilon}
\newcommand{\eps}{\epsilon}
\newcommand{\Pf}{\mbox{Pf}}

\newcommand{\A}{\mbox{Arf}}
\newcommand{\pin}{\mathrm{pin}^-}
\newcommand{\id}{\mathrm{id}}

\begin{document}

\title{Dimers on graphs in non-orientable surfaces}

\author{David Cimasoni}   
\address{ETH Zurich, Departement Mathematik, R\"amistrasse 101, 8092 Z\"urich, Switzerland}
\email{david.cimasoni@math.ethz.ch}

\begin{abstract}
The main result of this paper is a Pfaffian formula for the partition function of the dimer model on any graph $\G$
embedded in a closed, possibly non-orientable surface $\SI$. This formula is suitable for computational purposes,
and it is obtained using purely geometrical methods. The key step in the proof consists of a correspondence between
some orientations on $\G$ and the set of $\pin$ structures on $\SI$. This generalizes (and simplifies) the results of a
previous paper \cite{C-RI}.
\end{abstract}

\maketitle

\section{Introduction}
\label{sec:intro}

A dimer configuration on a graph $\G$ is a choice of a family of edges of $\G$, called dimers, such that each vertex of
$\G$ is adjacent to exactly one dimer. Assigning weights to the edges of $\G$ allows to define a probability measure
on the set of dimer configurations. The study of this measure has undergone spectacular advances in the past few years
(see in particular \cite{KenOk,KenOkS})
but one of the fundamental results on which these rely dates back from the early 60's. Back then, P. W. Kasteleyn
\cite{Kast0,Kast1} showed that the partition function of the dimer model on a planar graph is equal to the Pfaffian of a signed-adjacency matrix, the signs being determined by an orientation of the edges of $\G$ called a Kasteleyn orientation. For the square lattice on the torus, Kasteleyn showed
that the partition function can be written as a linear combination of 4 Pfaffians, corresponding to the 4 (equivalence
classes of) Kasteleyn orientations on such a graph. In the general case of a graph embedded in an orientable surface
of genus $g$, there are exactly $2^{2g}$ equivalence classes of such orientations, and Kasteleyn stated that
the partition function can be written as a linear combination of $2^{2g}$ Pfaffians
\cite{Kast1,Kast2}. A combinatorial proof of this fact for all oriented
surfaces was first obtained much later by Galluccio-Loebl \cite{L} and Tesler \cite{T}, independently.
(See also \cite{Dol}.)

The number of equivalence classes of Kasteleyn orientations on a graph $\G$ embedded in $\Sigma$ is also equal to the
number of equivalence classes of spin structures on $\SI$. An explicit construction relating a spin structure on a
surface with a Kasteleyn orientation on a graph with dimer configuration was suggested in
\cite{Kup}. In \cite{C-RI}, N. Reshetikhin and the author investigated further the relation between Kasteleyn
orientations and spin structures (see Theorem~\ref{thm:corr} and Corollary~\ref{cor:corr} below).
We also used this relation together with the identification of spin structures with quadratic forms to give a purely
geometric proof of the Pfaffian formula for closed oriented surfaces. Our final formula can be roughly expressed as follows:
given a graph $\G$ embedded in a closed oriented surface $\SI$ of genus $g$, the partition function
of the dimer model on $\G$ is given by
\[
Z=\frac{1}{2^g}\sum_{\xi\in\mathrm{Spin}(\SI)}(-1)^{\A(\xi)}\Pf(A^\xi),
\]
where $\mathrm{Spin}(\SI)$ denotes the set of equivalence classes of spin structures on $\SI$, $\A(\xi)\in\Z_2$ is the
Arf invariant of the spin structure $\xi$, and $A^\xi$ is the signed-adjacency matrix given by the Kasteleyn orientation corresponding to $\xi$. (See Theorem~\ref{thm:Pf} below for the precise statement.)

This formula obviously does not hold for graphs embedded in non-orientable surfaces, as neither spin structures
nor Kasteleyn orientations make sense in this setting. And yet, the dimer model on such graphs has been
the focus of some research, leading to Pfaffian-type formulae in several special cases. (See e.g. \cite{L-W1,L-W2} where the authors study square lattices in the M\"obius band and the Klein bottle.)

The purpose of the present paper is to extend our geometric approach of the dimer model to any graph embedded in (possibly)
non-orientable surfaces. The main idea is to replace spin structures by $\pin$ structures on $\SI$,
and to find a natural correspondence between these $\pin$ structures and some orientations on $\G\subset\SI$
(that we also call Kasteleyn orientations) -- see Theorem~\ref{thm:n-corr} and Corollary~\ref{cor:n-corr}.
We then make use of the identification of $\pin$ structures on $\SI$ with quadratic enhancements of the intersection form on
$H_1(\SI;\Z_2)$ to obtain the Pfaffian formula. It can be expressed as follows: given a graph $\G$ embedded in a closed possibly non-orientable surface
$\SI$,
\[
Z=\frac{1}{2^{b_1/2}}\sum_{\eta\in\mathrm{Pin}^{-}(\SI)}\exp(i\pi/4)^{\beta(\eta)}\Pf(A^\eta),
\]
where $b_1=\dim H_1(\SI;\Z_2)$, $\mathrm{Pin}^{-}(\SI)$ denotes the set of equivalence classes of $\pin$ structures on $\SI$,
$\beta(\eta)\in\Z_8$ is the Brown invariant of the $\pin$ structure $\eta$, and $A^\eta$ is the matrix given by the Kasteleyn orientation corresponding to $\eta$. (See Theorem~\ref{thm:Pf-n} below.)

This formula is mostly interesting from a theoretical point of view, but does not
seem very convenient for computational purposes. To this end, we obtain the following more usable result.
Let $\G$ be a graph embedded in a closed possibly non-orientable surface $\SI$ such that $\SI\setminus\G$ consists of
open 2-discs.  Recall that such a surface is of the form $\SI_g$, $\SI_g\#\R P^2$ or $\SI_g\#\K$, where $\SI_g$ denotes
the orientable surface of genus $g$ and $\K$ the Klein bottle. Let $\{\alpha_i\}$ be a set of simple closed curves on $\SI$,
transverse to $\G$, whose classes form a basis of $H_1(\SI_g;\Z_2)\subset\ho$.
If $\SI=\SI_g\#\K$, fix two disjoint simple closed curves $\beta_1,\beta_2$ on $\SI$, transverse to $\G$, disjoint from the
$\alpha_i$'s, whose classes form a basis of $H_1(\K;\Z_2)$ in $\ho$. Let $K$ be some well-chosen Kasteleyn orientation on
$\G\subset\SI$ (see Theorems~\ref{thm:Pf'} and \ref{thm:Pf'-n}), and for any $\eps=(\eps_1,\dots,\eps_{2g})\in\Z_2^{2g}$,
let $K_\eps$ denote the orientation obtained from $K$ as follows: invert the orientation $K$ on the edge $e$ of $\G$ each time $e$ intersects $\alpha_i$ with $\eps_i=1$. Finally, if $\SI=\SI_g\#\K$,
let $K'_{\eps}$ be obtained by inverting $K_\eps$ on $e$ each time the edge $e$ intersects $\beta_1$.

\begin{thmintro}
The partition function of the dimer model on $\Gamma$ is given by
\[
Z=\frac{1}{2^{g}}\Big|\sum_{\eps\in\Z_2^{2g}}(-1)^{\sum_{i<j}\eps_i\eps_j\alpha_i\cdot\alpha_j}
\Big(\mathrm{Re}(\Pf(A^{K_\eps}))+\mathrm{Im}(\Pf(A^{K_\eps}))\Big)\Big|,
\]
if $\SI=\SI_g$ or $\SI_g\#\R P^2$, and by
\[
Z=\frac{1}{2^{g}}\Big|\sum_{\eps\in\Z_2^{2g}}(-1)^{\sum_{i<j}\eps_i\eps_j\alpha_i\cdot\alpha_j}
\Big(\mathrm{Im}(\Pf(A^{K_\eps}))+\mathrm{Re}(\Pf(A^{K'_\eps}))\Big)\Big|,
\]
if $\SI=\SI_g\#\K$, where $A^{K_\eps}$ is the (complex-valued) weighted-adjacency matrix associated to the orientation
$K_\eps$.
\end{thmintro}

It should be mentioned that Tesler's combinatorial method \cite{T} is also valid for graphs in non-orientable surfaces.
Moreover, his approach yields an algorithm of the same complexity as ours, as both require the computation of
$2^{2-\chi(\Sigma)}$ Pfaffians of dimension the number of vertices of $\Gamma$. However, Tesler's final result
consists of an algorithmic way to compute the partition function, but not in a closed formula as the one obtained here.

Let us conclude this introduction with one last remark. It is well known that the partition function of the Ising model
on a graph $\Gamma\subset\Sigma$ can be expressed as the partition function of the dimer model on another graph
embedded in the same surface $\Sigma$. Therefore, our results could be used to compute the partition
function of the Ising model on graphs in non-orientable surfaces.
(See \cite{C-P} for an attempt to solve this problem using transfer matrices.)
However, we do not address this question in the present article.

\medskip

The paper is organized as follows. In Section~\ref{sec:dimers}, we introduce the dimer model and review Kasteleyn's
theory. For didactical reasons, we devote Section~\ref{sec:orientable} to the orientable case: we first recall the main
ideas of \cite{C-RI}, then present a greatly simplified correspondence between spin structures and Kasteleyn orientations.
We also prove both versions of the Pfaffian formula stated above, in the orientable case.
Our hope is that the reader will benefit from this warm up case before moving on to the core of the paper,
which lies in Sections~\ref{sec:Kast-n} to \ref{sec:Pf-n}.
There, we first extend the definition of a Kasteleyn orientation to graphs embedded in non-orientable surfaces
(Section~\ref{sec:Kast-n}), then show that these correspond naturally to $\pin$ structures on the surface
(Section~\ref{sec:pin}), and eventually prove the Pfaffian formulae (Section~\ref{sec:Pf-n}).

\section{Dimers and Pfaffians: Kasteleyn's theory}
\label{sec:dimers}

Let $\Gamma$ be a finite connected graph. A {\em dimer configuration\/} (or {\em perfect matching\/}) on $\G$ is a choice of edges of $\G$, called {\em dimers\/}, such that each vertex of $\G$ is adjacent to exactly one of these edges.
We shall denote by $\D(\G)$ the set of dimer configurations on $\G$. An {\em edge weight system\/} on $\G$
is a positive real-valued function $\mathrm{w}$ on the set of edges of $\G$. Such a system defines a
probability distribution on $\D(\G)$ by
\[
\mbox{Prob}(D)=\frac{\mathrm{w}(D)}{Z},
\]
where $\mathrm{w}(D)=\prod_{e\in D}\mathrm{w}(e)$ and
\[
Z=\sum_{D\in\D(\G)}\mathrm{w}(D).
\]
This probabilistic measure is the {\em Gibbs measure\/} for the dimer model on the graph $\G$ with weight system
$\mathrm{w}$, and $Z$ is the associated {\em partition function\/}.

\medskip

Given a fixed edge-weighted graph $\G$, the aim is to compute the associated partition function. Note that if
$\mathrm{w}$ is everywhere equal to one, this amounts to computing the number of dimer configurations on $\G$.

Kasteleyn's method is based on the following beautifully simple computation. If there exists a dimer configuration,
then the number of vertices of $\G$ is even. Enumerate them by $1,2,\dots,2n$, and fix an orientation $K$ of the edges
of $\G$. Let $A^K=(a_{ij}^K)$ denote the associated weighted skew-adjacency matrix; this is the $2n\times 2n$
skew-symmetric matrix whose coefficients are given by
\[
a^K_{ij}=\sum_{e}\e_{ij}^K(e)\mathrm{w}(e),
\]
where the sum is on all edges $e$ in $\Gamma$ between the vertices $i$ and $j$, and
\[
\e^K_{ij}(e)=
\begin{cases}
\phantom{-}1 & \text{if $e$ is oriented by $K$ from $i$ to $j$;} \\
-1 & \text{otherwise.}
\end{cases}
\]
Recall that the Pfaffian of a skew-symmetric matrix $A=(a_{ij})$ of size $2n$ is given by
\[
\Pf(A)=\sum_{[\sigma]\in\Pi}(-1)^\sigma a_{\sigma(1)\sigma(2)}\cdots a_{\sigma(2n-1)\sigma(2n)},
\]
where the sum is on the set $\Pi$ of matchings of $\{1,\dots,2n\}$, $\sigma$ is a permutation of $\{1,\dots,2n\}$
representing the matching $[\sigma]$, and $(-1)^\sigma\in\{\pm 1\}$ denotes the signature of $\sigma$. In the case of $A^K$,
a matching of $\{1,\dots,2n\}$ contributes to the Pfaffian if and only if it is realized by a dimer configuration
on $\G$, and this contribution is $\pm \mathrm{w}(D)$. More precisely,
\begin{equation}
\label{equ:Pf}
\Pf(A^K)=\sum_{D\in\D(\G)}\e^K(D)\mathrm{w}(D),
\end{equation}
where the sign $\e^K(D)$ can be computed as follows: if the dimer configuration $D$ is given by edges $e_1,\dots,e_n$ matching vertices $i_\ell$ and $j_\ell$ for $\ell=1,\dots,n$, let $\sigma$ denote the permutation sending
$(1,\dots, 2n)$ to $(i_1,j_1,\dots,i_n,j_n)$, and set
\begin{equation}
\label{equ:eps}
\e^K(D)=(-1)^\sigma\prod_{\ell=1}^n\e^K_{i_\ell j_\ell}(e_\ell).
\end{equation}
The problem of expressing $Z$ as a Pfaffian now boils down to finding an orientation $K$ of the edges of $\G$
such that $\e^K(D)$ does not depend on $D$.

Obviously, any dimer configuration $D$ can be considered as a cellular 1-chain $D\in C_1(\G;\Z_2)$ such that
$\partial D=\sum_{v}v$, the sum being on all vertices of $\G$.
Hence, given any two dimer configurations $D,D'$, their sum $D+D'$ is a 1-cycle.
The connected components of this 1-cycle are disjoint simple loops of even length; let us denote them by
$\{C_i\}_i$. An easy computation shows that
\begin{equation}
\label{equ:cc}
\e^K(D)\e^K(D')=\prod_i(-1)^{n^K(C_i)+1},
\end{equation}
where $n^K(C_i)$ denotes the number of edges of $C_i$ where a fixed orientation of $C_i$ differs from $K$.
(Since $C_i$ has even length, the parity of this number is independent of the orientation of $C_i$.)
Therefore, we are now left with the problem of finding an orientation $K$ of $\G$ such that, for any cycle $C$ of even
length such that $\G\setminus C$ admits a dimer configuration, $n^K(C)$ is odd. Such an orientation was called
{\em admissible\/} by Kasteleyn; nowadays, the term of {\em Pfaffian orientation\/} is commonly used.
By the discussion above, if $K$ is a Pfaffian orientation, then $Z=|\Pf(A^K)|$.

Kasteleyn's early triumph was to prove that every planar graph admits a Pfaffian orientation. More precisely, let
$\G$ be a graph embedded in the plane. Each face $f$ of $\G\subset\R^2$ inherits the (say, counterclockwise)
orientation of $\R^2$, so $\partial f$ can be oriented as the boundary of the oriented face $f$.

\begin{Kast}[\cite{Kast1,Kast2}]
Given $\G\subset\R^2$, there exists an orientation $K$ of $\G$ such that, for each face $f$
of $\G\subset\R^2$, $n^K(\partial f)$ is odd. Furthermore, such an orientation is Pfaffian.
\end{Kast}

An amazing consequence of this result is that it enables to compute the partition function of the dimer model
on a planar graph in polynomial time.

There is no hope to extend this result to the general case. Indeed, some graphs (such as the complete bipartite graph
$K_{3,3}$) do not admit a Pfaffian orientation. More generally, enumerating the dimer configurations on a graph
is a $\#P$-complete problem \cite{Val}. It turns out that Kasteleyn's method does extend to surfaces, but one needs to compute
many Pfaffians. This is the aim of the following section.

\section{The Pfaffian formula for graphs on orientable surfaces}
\label{sec:orientable}

In \cite{C-RI}, N. Reshetikhin and the author derived a Pfaffian formula for graphs embedded in closed orientable surfaces.
However, the central argument -- that is, the correspondence between Kasteleyn
orientations and spin structures -- was quite intricate. Also, the Pfaffian formula did not appear to be very convenient
to use in practice. In this section, we shall present a more transparent
correspondence and recall the other main steps of the proof. We shall also give another version of the Pfaffian
formula, more suitable for computational purposes (see Theorem~\ref{thm:Pf'}).
As stated in the introduction, our hope is that the reader will benefit from this warm up
case before moving on to the more involved case presented in Sections~\ref{sec:Kast-n} to \ref{sec:Pf-n}.

\subsection{Kasteleyn orientations}
\label{sub:Kast}

Throughout this section, $\SI$ will denote a closed connected surface endowed with an orientation that will be pictured
counterclockwise.
By a {\em surface graph\/}, we mean a graph $\G$ embedded in $\SI$ as the 1-skeleton of a cellular decomposition $X$
of $\SI$. This simply means that the complement of $\G$ in $\SI$ consists of open 2-discs.
We shall use the same notation $X$ for the surface graph and the cell complex realizing it. Note that any
finite connected graph can be realized as a surface graph.

An orientation $K$ of the 1-cells of a surface graph $X$ is called a {\em Kasteleyn orientation on $X$\/} if,
for each 2-cell $f$ of $X$, the following condition holds: the number $n^K(\partial f)$ of edges in $\partial f$
where $K$ disagrees with the orientation on $\partial f$ induced by the counterclockwise orientation on $f$, is odd.
Given a Kasteleyn orientation on $X$, there is an obvious way to obtain another one: pick a vertex of $X$ and flip
the orientation of all the edges adjacent to it. Two Kasteleyn orientations are said to be {\em equivalent\/} if
they can be related by such moves. Let us denote by $\K(X)$ the set of equivalence classes of Kasteleyn
orientations on $X$. 

\begin{proposition}
\label{prop:Kast}
A surface graph $X$ admits a Kasteleyn orientation if and only if $X$ has an even number of vertices.
In this case, the set $\K(X)$ is an $H^1(\SI;\Z_2)$-torsor, that is, it admits a freely transitive action of
the group $H^1(\SI;\Z_2)$.\qed
\end{proposition}

The proof can be found either in \cite[Section 4]{C-RI}, or in Section~\ref{sec:Kast-n} of the present paper where the
more general Theorem~\ref{thm:Kast-n} is proved.
This proposition implies that, if $X$ has an even number of vertices, then it admits exactly $2^{2g}$ equivalence classes
of Kasteleyn orientations, where $g$ is the genus of $\SI$. It actually also implies the following.

\begin{corollary}
A surface graph of genus $g$ with an even number $V$ of vertices admits exactly $2^{2g+V-1}$ Kasteleyn orientations.
\end{corollary}
\begin{proof}
By Proposition~\ref{prop:Kast}, it is sufficient to prove that, given any orientation $K_0$, the number of orientations
equivalent to $K_0$ is equal to $2^{V-1}$. Let $\mathcal{P}(X^0)$ denote the set of subsets of the vertices of $X$,
and let $\varphi\colon\mathcal{P}(X^0)\to\{K\,|\,K\sim K_0\}$ be given by $\varphi(S)=K$, the orientation
obtained from $K_0$ by changing the orientation around all vertices of $S$. The map $\varphi$ is obviously surjective.
Using the fact that $X$ is connected, one easily checks that $\varphi(S)=\varphi(S')$ if and only if $S=S'$ or
$S$ and $S'$ form a partition of $X^0$. Hence, $\varphi$ is two to one, proving the claim.
\end{proof}

\subsection{Discrete spin structures}
\label{sub:spin}

Let us now recall several general facts about spin structures on a compact oriented surface $\SI$, referring to
\cite[Section 5.2]{C-RI} for details. We assume throughout that $\SI$ is endowed with a fixed Riemannian metric.

First of all, it is well known that
the set $\mathrm{Spin}(\SI)$ of equivalence classes of spin structures on $\SI$ is an $H^1(\SI;\Z_2)$-torsor. Also,
any spin structure can be described by a vector field on $\SI$ with isolated zeroes of even index; conversely,
any such vector field defines a spin structure.
Finally, a theorem of D. Johnson \cite{Jo} asserts the existence of
an $H^1(\SI;\Z_2)$-equivariant bijection from $\mathrm{Spin}(\SI)$ onto the set $\Q(\SI)$
of quadratic forms on $H_1(\SI;\Z_2)$.
(Recall that such a form is a map $q\colon H_1(\SI;\Z_2)\to\Z_2$ such that $q(x+y)=q(x)+q(y)+x\cdot y$ for all
$x,y\in H_1(\SI;\Z_2)$, where $\cdot$ denotes the intersection form. The set $\Q(\SI)$ is clearly
an $H^1(\SI;\Z_2)$-torsor.) More explicitly, given a spin structure $\xi\in\mathrm{Spin}(\SI)$ and a vector field $Y$
representing it, then the associated quadratic form $q_\xi\colon H_1(\SI;\Z_2)\to\Z_2$ is defined as follows.
Represent a class $\alpha\in H_1(\SI;\Z_2)$ by a collection of disjoint oriented regular simple closed curves
$C_1,\dots,C_m$ in $\SI$ avoiding the zeroes of the vector field $Y$, and set
\[
q_\xi(\alpha)= \sum_{i=1}^m (w_{\overrightarrow{C_i}}(Y)+1) \pmod{2}
\]
where $w_{\overrightarrow{C_i}}(Y)$ denotes the winding number of the vector field $Y$ with respect to
the tangential vector field along $C_i$.
\medskip

Now, the game consists in trying to encode combinatorially a spin structure on a surface $\SI$, or
equivalently, a vector field on $\SI$ with isolated zeroes of even index. Let us begin by fixing a cellular decomposition
$X$ of $\SI$.

$\bullet$
To construct a (unit length) vector field along the 0-skeleton $X^0$, we just need to specify one tangent
direction at each vertex of $X$. Such an information is given by a dimer configuration $D$ on $X^1$: at each vertex,
point in the direction of the adjacent dimer.

$\bullet$
This vector field along $X^0$ extends to a unit vector field on $X^1$, but not uniquely. Roughly speaking,
it extends in two different natural ways along each edge of $X^1$, depending on the sense of rotation of the resulting
vector field. We shall
encode this choice by an orientation $K$ of the edges of $X^1$, together with the following convention: moving along
an oriented edge, the tangent vector first rotates counterclockwise until it points in the direction of the edge,
then rotates clockwise until it points backwards, and finally rotates counterclockwise until it coincides with
the tangent vector at the end vertex. This is illustrated in Figure~\ref{fig:vector}.
\begin{figure}[htbp]
\labellist\small\hair 2.5pt
\pinlabel {$K$} at 220 165
\pinlabel {$D$} at 75 60
\pinlabel {$D$} at 395 190
\endlabellist
\centerline{\psfig{file=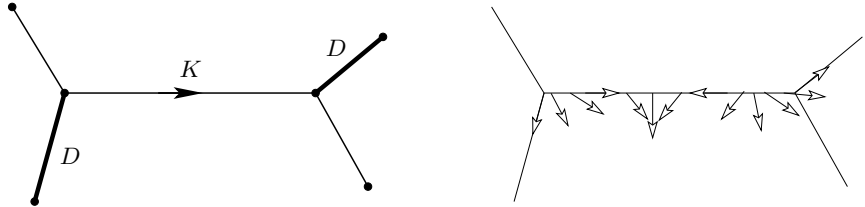,height=2.7cm}}
\caption{Construction of the vector field along the 1-skeleton of $X$.}
\label{fig:vector}
\end{figure}

$\bullet$
Each face of $X$ being homeomorphic to a 2-disc, the unit vector field defined along $X^1$ naturally extends to
a vector field $Y_D^K$ on $X$, with one isolated zero in the interior of each face. Recall that in order for this vector
field to define a spin structure, each zero needs to be of even index. The following lemma implies that this is the
case if and only if $K$ is a Kasteleyn orientation.

\begin{lemma}
\label{lemma:index}
For each face $f$ of $X$, the index of the zero of $Y_D^K$ in $f$ has the parity of $n^K(\partial f)+1$.
\end{lemma}
\begin{proof}
By definition, the index of this zero is equal to $w_{\partial f}(Y_D^K)$, the winding number of the vector
field $Y_D^K$ along $\partial f$ (with respect to a constant vector field along $f$). If $K'$ is obtained from $K$ by
changing the orientation along one edge $e$ of $\partial f$, then $w_{\partial f}(Y_D^{K'})-w_{\partial f}(Y_D^K)$ is
equal to the winding number of the vector field obtained as follows: first move along $e$ in one direction
as described in Figure~\ref{fig:vector}, and then back along $e$ in the opposite direction using again the construction
of Figure~\ref{fig:vector}. Using the fact that the vector field
is determined at vertices by a dimer configuration, one easily checks that this winding number is equal to $\pm 1$.
Since the right-hand side of the equation $w_{\partial f}(Y_D^K)=n^K(\partial f)+1\pmod{2}$ also changes
when replacing $K$ by $K'$, it may be assumed that $n^K(\partial f)=0$. In this case, $K$ orients each edge of
$\partial f$ counterclockwise around $f$, and the resulting vector field is isotopic to the vector field pointing
outwards along $\partial f$. The winding number of this vector field along $\partial f$ being equal to $1$,
the equality is proved.
\end{proof}

Hence, a dimer configuration $D$ on $X^1$ and a Kasteleyn orientation $K$ on $X$ determine a spin
structure on $\SI$, i.e: a quadratic form $q^K_D\colon H_1(\SI;\Z_2)\to\Z_2$. It can be computed as follows.

\begin{lemma}
\label{lemma:quadratic}
The quadratic form $q^K_D\colon H_1(\SI;\Z_2)\to\Z_2$ is characterized by the following property: if $C$ is an oriented
simple closed curve on $X^1$, then
\[
q^K_D([C])=n^K(C)+\ell_D(C)+1\pmod{2},
\]
where $\ell_D(C)$ denotes the number of vertices in $C$ such that the adjacent dimer of $D$ points out to the left of
$C$.
\end{lemma}
\begin{proof}
Since $H_1(\SI;\Z_2)=H_1(X;\Z_2)$, any element in $H_1(\SI;\Z_2)$ can be
represented by a 1-cycle in $X^1$. The map $q^K_D$ being a quadratic form, its value on such a cycle is determined by
its value on simple closed curves in $X^1$. Let $C$ be such a curve.
By construction, $q^K_D([C])$ is equal to $w_{\overrightarrow{C}}(Y^K_D)+1$, so we are left with the proof that
$w_{\overrightarrow{C}}(Y^K_D)\equiv n^K(C)+\ell_D(C)\pmod{2}$. By the argument given at the beginning of the proof of Lemma
\ref{lemma:index}, the parity of $w_{\overrightarrow{C}}(Y^K_D)$ changes when $K$ is inverted along one edge of $C$.
Since the same obviously holds for $n^K(C)+\ell_D(C)$, it may be assumed that $n^K(C)=0$.
Furthermore, the parity of $w_{\overrightarrow{C}}(Y^K_D)$ also changes when a dimer pointing out to the left of
$C$ is replaced by a dimer either on $C$, or pointing out to the right of $C$. Hence, it may be assumed that
$\ell_D(C)=0$ as well. But in this case, $Y^K_D$ is isotopic to the vector field pointing constantly to the right
of $C$, so $w_{\overrightarrow{C}}(Y^K_D)=0$. This proves the lemma.
\end{proof}

We can now state our correspondence theorem. It is a straightforward consequence of Lemma~\ref{lemma:quadratic}
(the same way -- in \cite{C-RI} -- Proposition 1 and Corollary 2 follow from Theorem 3).

\begin{theorem}
\label{thm:corr}
Let $X$ be a cellular decomposition of an oriented closed surface $\SI$. Then, any dimer configuration $D\in\D(X^1)$
induces an $H^1(\SI;\Z_2)$-equivariant bijection
\[
\psi_D\colon\K(X)\to\Q(\SI)=\mathrm{Spin}(\SI),\quad [K]\mapsto q_D^K
\]
from the set of equivalence classes of Kasteleyn orientations on $X$ to the set of equivalence classes of spin structures
on $\SI$. Furthermore, given another dimer configuration $D'\in\D(X^1)$, $\psi_{D'}$ is obtained from $\psi_D$ by action of the Poincar\'e dual to
$[D+D']\in H_1(\SI;\Z_2)$. \qed
\end{theorem}

As stated above, the correspondence depends on the choice of $D\in\D(X^1)$. This can be remedied as follows. Let
$\mathcal{B}=\{\alpha_i\}$ denote a family of closed curves in $\SI$, transverse to $X^1$, whose classes
form a basis of $H_1(\SI;\Z_2)$. Given any $D\in\D(X^1)$, let $\varphi_\mathcal{B}^D\in\coho=\Hom(\ho,\Z_2)$ be given by
$\varphi_\mathcal{B}^D([\alpha_i])=\alpha_i\cdot D$ for $1\le i\le 2g$. Finally, let $q^K_\mathcal{B}\in\Q(\SI)$
be defined by $q^K_\mathcal{B}=q^K_D+\varphi_\mathcal{B}^D$.

\begin{corollary}
\label{cor:corr}
Let $X$ be a cellular decomposition of an oriented closed surface $\SI$ such that $X^1$ admits a dimer configuration
$D$. Then, the map
\[
\psi_\mathcal{B}\colon\K(X)\to\Q(\SI)=\mathrm{Spin}(\SI),\quad [K]\mapsto q_\mathcal{B}^K
\]
is an $H^1(\SI;\Z_2)$-equivariant bijection which does not depend on $D$.
\end{corollary}
\begin{proof}
This map is obviously an $H^1(\SI;\Z_2)$-equivariant bijection, as it is obtained from $\psi_D$ via translation by
$\varphi_\mathcal{B}^D\in\coho$. Furthermore, given $D,D'\in\D(X^1)$,
\[
(\varphi_\mathcal{B}^D+\varphi_\mathcal{B}^{D'})([\alpha_i])=\alpha_i\cdot(D+D')=[\alpha_i]\cdot[D+D'].
\]
In other words, $\varphi_\mathcal{B}^D+\varphi_\mathcal{B}^{D'}$ is equal to $[D+D']^*$, the Poincar\'e dual to
$[D+D']$. Since $\psi_{D'}=\psi_{D}+[D+D']^*$, it follows that $\psi_\mathcal{B}=\psi_D+\varphi_\mathcal{B}^{D}$
does not depend on $D$.
\end{proof}

This whole paragraph shows that Kasteleyn orientations should be understood as ``discrete spin structures'' on surfaces.
This terminology is already present in the literature: according to Mercat \cite{Mer}, a discrete spin structure on $X$
is a double cover $p\colon\widetilde{X}^1\to X^1$ whose restriction to $p^{-1}(\partial f)$ is the non-trivial double
cover, for any face $f$ of $X$. As Mercat points out, such a cover is encoded by the homomorphism
$\mu\colon Z_1(X)\to\Z_2$ which maps a 1-cycle $C\subset X^1$ to $\mu(C)=0$ if and only if $C$ lifts to a cycle
$\widetilde C$ in $\widetilde{X}^1$. The correspondence with our point of view is straightforward: given a Kasteleyn
orientation $K\in\K(X)$, the associated $\mu_K$ is simply given by $\mu_K(C)=n^K(C)$.

\subsection{The Pfaffian formula}
\label{sub:Pfaff}

The aim of the previous paragraph was to give a natural correspondence between Kasteleyn orientations and
spin structures. But as a direct consequence of Lemma~\ref{lemma:quadratic}, we also obtain immediately the following
non-trivial combinatorial result.

\begin{proposition}
\label{prop:quadratic}
Let $K$ be a Kasteleyn orientation on $X$, and $D$ be a dimer configuration on $X^1$. Given a homology class
$\alpha\in H_1(\SI;\Z_2)$, represent it by oriented simple closed curves $C_1,\dots,C_m$ in $X^1$. Then, the equality
\[
q_D^K(\alpha)=\sum_{i=1}^m(n^K(C_i)+\ell_D(C_i)+1)+\sum_{1\le i<j\le m}C_i\cdot C_j\pmod{2}
\]
determines a well-defined quadratic form $q_D^K\colon H_1(\SI;\Z_2)\to\Z_2$.\qed
\end{proposition}

We shall now use this combinatorial information, together with the results and notation of Section~\ref{sec:dimers},
to derive our Pfaffian formula.

Let $\G$ be a finite connected graph endowed with an edge weight system $\mathrm{w}$. If $\G$ does not admit any dimer
configuration, then the partition function $Z$ is obviously zero. So, let us assume that $\G$ admits a dimer
configuration $D_0$. Enumerate the vertices of $\G$ by $1,2,\dots,2n$ and embed $\G$ in a closed orientable surface
$\SI$ of genus $g$ as the 1-skeleton of a cellular decomposition $X$ of $\SI$.

Since $\G$ has an even number of vertices, the set $\K(X)$ is an $H^1(\SI;\Z_2)$-torsor. For any Kasteleyn orientation
$K$, the Pfaffian of the associated weighted skew-adjacency matrix satisfies
\begin{align*}
\e^K(D_0)\Pf(A^K)&\overset{(\ref{equ:Pf})}{=}\sum_{D\in\D(\G)}\e^K(D_0)\e^K(D)\,\mathrm{w}(D)\\
	&\overset{(\ref{equ:cc})}{=}\sum_{D\in\D(\G)}(-1)^{\sum_i (n^K(C_i)+1)}\mathrm{w}(D),
\end{align*}
where the $C_i$'s are the connected components of the cycle $D+D_0\in C_1(X;\Z_2)$.
Note that given any vertex of $C_i$, the adjacent dimer of $D_0$ lies on $C_i$, so that $\ell_{D_0}(C_i)=0$.
Since the cycles $C_i$ are disjoint, Proposition~\ref{prop:quadratic} gives
\[
\sum_i(n^K(C_i)+1)=\sum_i(n^K(C_i)+\ell_{D_0}(C_i)+1)=q^K_{D_0}([D+D_0]).
\]
Therefore, every element $[K]$ of $\K(X)$ induces a linear equation
\[
\e^K(D_0)\Pf(A^K)=\sum_{\alpha\in H_1(\SI;\Z_2)}(-1)^{q^K_{D_0}(\alpha)}Z_\alpha(D_0),
\]
where $Z_\alpha(D_0)=\sum_{[D+D_0]=\alpha}\mathrm{w}(D)$, the sum being over all $D\in\D(\G)$ such that $[D+D_0]=\alpha$.
It is an easy exercise to solve this linear system of $2^{2g}$ equations with $2^{2g}$ unknowns, and to 
obtain the following formula for $Z=\sum_\alpha Z_\alpha(D_0)$. (See \cite[Theorem 5]{C-RI} for details.)

\begin{theorem}
\label{thm:Pf}
Let $\G$ be a graph embedded in a closed oriented surface $\SI$ of genus $g$ such that $\SI\setminus\G$ consists of
open 2-discs. Then, the partition function of the dimer model on $\Gamma$ is given by the formula
\[
Z=\frac{1}{2^{g}}\sum_{[K]\in\K(X)}(-1)^{\A(q^{K}_{D_0})}\e^K(D_0)\Pf(A^{K}),
\]
where the sum is taken over all equivalence classes of Kasteleyn
orientations, and $\A(q)\in\Z_2$ denotes the Arf invariant of the quadratic form $q$.
Furthermore, the sign $(-1)^{\A(q^{K}_{D_0})}\e^K(D_0)$ does not depend on $D_0$.\qed
\end{theorem}

Note that for a fixed $D_0$, and for any element of $\K(X)$, one can always choose a Kasteleyn orientation in
this equivalence class such that $\e^K(D_0)=1$. By Theorem~\ref{thm:corr}, this leads to the formula stated in the
introduction:
\[
Z=\frac{1}{2^g}\sum_{\xi\in\mathrm{Spin}(\SI)}(-1)^{\A(\xi)}\Pf(A^\xi),
\]
where $A^\xi$ is the matrix $A^K$ for any Kasteleyn orientation $K$ such that $q^K_{D_0}=\xi$ and $\e^K(D_0)=1$.

This formula is reminiscent of \cite[Equation 6.9]{AMV}, drawing a strong analogy between the dimer model on $\G$
and the bosonic Quantum Field Theory on the compact Riemann surface $\SI$. However, as it stands here, it
is not very convenient for computational purposes. Indeed, it seems to require the choice of a dimer configuration $D_0$,
which is often in practice very hard -- if not impossible -- to find. Also, the computation of each quadratic form can be
very tedious. Nevertheless, we shall now show that this formula can actually be used in a very efficient way to
compute the partition function $Z$.

Let $\mathcal{B}=\{\alpha_i\}$ be a set of simple closed curves on $\SI$,
transverse to $\G$, whose classes form a basis of $\ho$. Fix a Kasteleyn orientation $K$ on $\G\subset\SI$
which satisfies the following property: for any $\alpha_i\in\mathcal{B}$, let $C_i$ denote the oriented 1-cycle in $\G$
having $\alpha_i$ to its immediate left, and meeting every vertex of $\G$ adjacent to $\alpha_i$ on this side.
We require $n^K(C_i)$ to be odd for each $i$.
(There are in fact two possible choices for $C_i$, corresponding to the two sides of $\alpha_i$, but the parity condition
above does not depend on which one is chosen.) Finally, for any $\eps=(\eps_1,\dots,\eps_{2g})\in\Z_2^{2g}$,
let $K_\eps$ denote the Kasteleyn orientation obtained from $K$ as follows: invert the orientation $K$ on the edge
$e$ of $\G$ each time $e$ intersects $\alpha_i$ with $\eps_i=1$.

\begin{theorem}
\label{thm:Pf'}
Let $\G$ be a graph embedded in a closed oriented surface $\SI$ of genus $g$ such that $\SI\setminus\G$ consists of
open 2-discs, and fix a set of simple closed curves $\{\alpha_i\}$ on $\SI$, transverse to $\G$, whose classes form a basis of $\ho$. Then, the partition function of the dimer model on $\Gamma$ is given by the formula
\[
Z=\frac{1}{2^{g}}\Big|\sum_{\eps\in\Z_2^{2g}}(-1)^{\sum_{i<j}\eps_i\eps_j\alpha_i\cdot\alpha_j}\Pf(A^{K_\eps})\Big|,
\]
where $K_\eps$ are the Kasteleyn orientations described above.
\end{theorem}
\begin{proof}
If $\G$ does not admit any dimer configuration, then Equation (\ref{equ:Pf}) implies that $\Pf(A^{K_\eps})=0$
for all $\eps$, and our equality holds. Therefore,
it may be assumed that there exists a $D\in\D(\G)$. In particular, $\G$ has an even number of vertices, so
$\K(X)$ is an $\coho$-torsor by Proposition~\ref{prop:Kast}. The set $\{K_\eps\}_{\eps\in\Z_2^{2g}}$ is constructed to
contain one element in each equivalence class of Kasteleyn orientations, so Theorem~\ref{thm:Pf} gives the equality
\begin{align*}
Z&=\frac{1}{2^{g}}\sum_{\eps\in\Z_2^{2g}}(-1)^{\A(q^{K_\eps}_D)}\e^{K_\eps}(D)\Pf(A^{K_\eps})\\
&=\frac{1}{2^{g}}\Big|
\sum_{\eps\in\Z_2^{2g}}(-1)^{\A(q^{K_\eps}_D)+\A(q^{K}_D)}\e^{K_\eps}(D)\e^{K}(D)\Pf(A^{K_\eps})\Big|.\tag{$\star$}
\end{align*}
By \cite[Lemma 1]{C-RI},
\[
\A(q^{K_\eps}_D)+\A(q^{K}_D)=q_D^K([\Delta_\eps]),
\]
where the Poincar\'e dual $[\Delta_\eps]^*$ of $[\Delta_\eps]\in\ho$ is required to satisfy 
$q^K_D+[\Delta_\eps]^*=q^{K_\eps}_D$. By Theorem~\ref{thm:corr}, this is equivalent to
$K+[\Delta_\eps]^*=K_\eps$. The very definition of $K_\eps$ implies that $\Delta_\eps=\sum_i\eps_i\alpha_i$
represents the right homology class. On the other hand, one easily checks the equality
\[
\e^{K_\eps}(D)\e^{K}(D)=(-1)^{\Delta_\eps\cdot D}.
\]
We thus obtain that the coefficient in $(\star)$ corresponding to $\eps$ is equal to
\[
(-1)^{q_D^K([\Delta_\eps])+\Delta_\eps\cdot D}=(-1)^{q^K_\mathcal{B}([\Delta_\eps])}
=(-1)^{q^K_\mathcal{B}(\sum_i\eps_i[\alpha_i])},
\]
using the notation of Corollary~\ref{cor:corr}.
Since $q^K_\mathcal{B}$ is a quadratic form,
\[
q^K_\mathcal{B}(\sum_i\eps_i[\alpha_i])=
\sum_i\eps_i q^K_\mathcal{B}([\alpha_i])+\sum_{i<j}\eps_i\eps_j\alpha_i\cdot\alpha_j.
\]
Therefore, it remains to check that $q^K_\mathcal{B}([\alpha_i])$ vanishes for all $i$.
To do so, consider the oriented closed curve $C_i$ in $\G$ having $\alpha_i$ to its immediate left,
and meeting every vertex of $\G$ adjacent to $\alpha_i$ on this side. Obviously, $C_i$ and $\alpha_i$ are homologous,
and by construction, $\ell_D(C_i)=\alpha_i\cdot D$. Therefore,
\[
q^K_\mathcal{B}([\alpha_i])=q^K_\mathcal{B}([C_i])=n^K(C_i)+\ell_D(C_i)+\alpha_i\cdot D+1=n^K(C_i)+1.
\]
We have chosen $K$ precisely so that this number is even for every $i$.
\end{proof}

\section{Kasteleyn orientations in the non-orientable case}
\label{sec:Kast-n}

We shall now generalize the methods and results of Section~\ref{sec:orientable} to the case of graphs embedded in (possibly)
non-orientable closed surfaces. Once again, all the concepts will be presented in an intrinsic way,
allowing us to give geometrical proofs with no combinatorial argument.

Let us begin with the generalization of the notion of Kasteleyn orientation.
Throughout this section, $\SI$ will designate a possibly non-orientable closed connected surface, X a cellular decomposition
of $\SI$, and $\G$ its 1-skeleton.

\subsection{Extension of the definition of a Kasteleyn orientation}

The definition of a Kasteleyn orientation on $X$ given in Section~\ref{sub:Kast} does
not make sense in the present setting, as the faces of $X$ are not oriented. We will hence work
in the orientation cover of $\SI$, that is, the 2-fold cover $\widetilde\SI\stackrel{\pi}{\to}\SI$ determined by the first Stiefel-Whitney class $w_1=w_1(\SI)\in\coho$ of $\SI$.
We shall denote by $\widetilde X$ the cellular decomposition of the orientable surface $\widetilde\SI$ induced by
$X$ and $\pi$.

The extension of the notion of a Kasteleyn orientation requires a labelling of the vertices of $\widetilde X$ with
signs, which is a little tedious to define intrinsically. Following our general approach, we will now state
this intrinsic definition, but the reader impatient to work with examples should replace this paragraph with
Remark~\ref{rem:Kast}.
Let us fix a 1-cocycle $\omega\in C^1(X;\Z_2)$ which represents $w_1$. This consists simply in a decomposition of
the edges of $\G$ into 0-edges and 1-edges, such that the local orientation of $\SI$ is preserved along a 1-cycle if and
only if this cycle contains an even number of 1-edges. The choice of such an $\omega$ determines a labelling of the vertices
of $\widetilde X$ with signs $\pm$'s such that each vertex of $X$ is covered by two vertices with opposite signs, and
$\omega(e)=0$ if and only if the two endpoints of a lift of $e$ have the same label. Note that this labelling is uniquely
determined by $\omega$ up to a global swap of all the signs. It induces an orientation on $\widetilde X$: simply
paste together a local orientation (say, counterclockwise) near the vertices labelled $+$ and the opposite (clockwise)
local orientation near the vertices labelled $-$.

Any orientation $K$ of the edges of $X$ lifts to an orientation $\widetilde K$ of the edges of $\widetilde X$.
Given a face $f$ of $X$, and a lift $\widetilde f$, consider the number
\[
c^K(\widetilde f)=
n^{\widetilde{K}}(\partial\widetilde f)+\#\{\text{edges in $\partial\widetilde f$ joining two vertices labelled $-$}\}+1,
\]
where $\partial\widetilde f$ is oriented as the boundary of the oriented face $\widetilde f$.
(As before, $n^{\widetilde{K}}(\partial\widetilde f)$ denotes the number of edges in $\partial\widetilde f$ where
$\widetilde{K}$ disagrees with the orientation of $\partial\widetilde f$.)
Using the fact that $\partial f$ contains an even number of 1-edges, one easily checks that the parity of
$c^K(\widetilde f)$ does not depend on the choice of the lift $\widetilde f$ of $f$.
For the same reason, the parity of $c^K(\widetilde f)$ is unchanged if one swaps all the signs of the vertices of
$\widetilde X$, as this also reverses the orientation of $\widetilde X$.

Therefore, the parity $c^K(f)\in\Z_2$ of the number $c^K(\widetilde f)$ only depends on $K$, $f$ and $\omega$.
By analogy with the orientable case, we shall call it the {\em Kasteleyn curvature of $K$ at $f$\/}.
An orientation $K$ is a {\em Kasteleyn orientation on $(X,\omega)$\/} if it has zero curvature.

As usual, we shall say that two orientations are {\em equivalent\/} if they can be obtained from each other by flipping
the edge orientations around a set of vertices. If $K$ is Kasteleyn, and $K'$ is equivalent to $K$, then $K'$ is obviously Kasteleyn. We shall denote by $\K(X,\omega)$ the set of equivalence classes of Kasteleyn orientation on $(X,\omega)$.

\begin{remark}
Recall that a surface $\SI$ is orientable if and only if $w_1$ vanishes.
In this case, $\omega=0$ provides a natural choice, and
a Kasteleyn orientation on $(X,0)$ is simply a Kasteleyn orientation on $X$ as defined in Section~\ref{sub:Kast}. Therefore,
$\K(X,0)$ is nothing but $\K(X)$. Once again, Sections~\ref{sec:Kast-n} to \ref{sec:Pf-n} should be understood as
a generalization of the previous one, which corresponds to the case $\omega=0$.
\end{remark}

\begin{remark}
\label{rem:Kast}
When working with examples, it is often convenient to represent the surface $\SI$ as a planar polygon $P$
with some pairs of sides identified, and to draw $\G$ in $P$ intersecting $\partial P$ transversally.
Such a representation of $X$ induces a natural cocycle $\omega\in C^1(X;\Z_2)$ representing the first
Stiefel-Whitney class of $\SI$: Let us call a side of $\partial P$ a {\em 1-side} if the corresponding identification
does not preserve the orientation of $P$. For an edge $e$ of $\G$, simply define $\omega(e)$ to be the parity of the
intersection number of $e$ with all the 1-sides of $P$.

For this $\omega$, it is easy to check whether a given orientation $K$ of $\G$ is Kasteleyn or not: Take one
counterclockwise-oriented copy of $P\supset\G$ with vertices labelled $+$, one clockwise-oriented copy of $P\supset\G$
with vertices labelled $-$, glue these two polygons along their 1-sides according to the prescribed identifications,
and make the remaining side identifications in each copy of $P$. The result is the oriented surface $\widetilde X$,
where one can compute the Kasteleyn curvature. Let us illustrate this on an example.
\end{remark}

\begin{example}
Let $P$ denote the model of the Klein bottle $\mathcal K$ given by a hexagone with sides identified according to the
word $a^2bc^2b^{-1}$. Note that the 1-sides are exactly the four sides of $\partial P$ corresponding to the letters
$a$ and $c$. Now, consider the square lattice $\G$ embedded in $P$ as illustrated in Figure~\ref{fig:Klein}, and let
$X$ denote the induced cellular decomposition of $\mathcal{K}$.
The graph $\G$ admits exactly six edges $e$ with $\omega(e)=1$. These are the six edges crossing the 1-sides of $\partial P$.
One easily checks that the orientation of the edges of $\G$ given by the arrows in Figure~\ref{fig:Klein} is
a Kasteleyn orientation on $(X,\omega)$.

\begin{figure}[htbp]
\labellist\small\hair 2.5pt
\pinlabel {$a$} at 330 420
\pinlabel {$a$} at 114 420
\pinlabel {$b$} at 465 220
\pinlabel {$b^{-1}$} at -30 220
\pinlabel {$c$} at 330 10
\pinlabel {$c$} at 114 10
\endlabellist
\centerline{\psfig{file=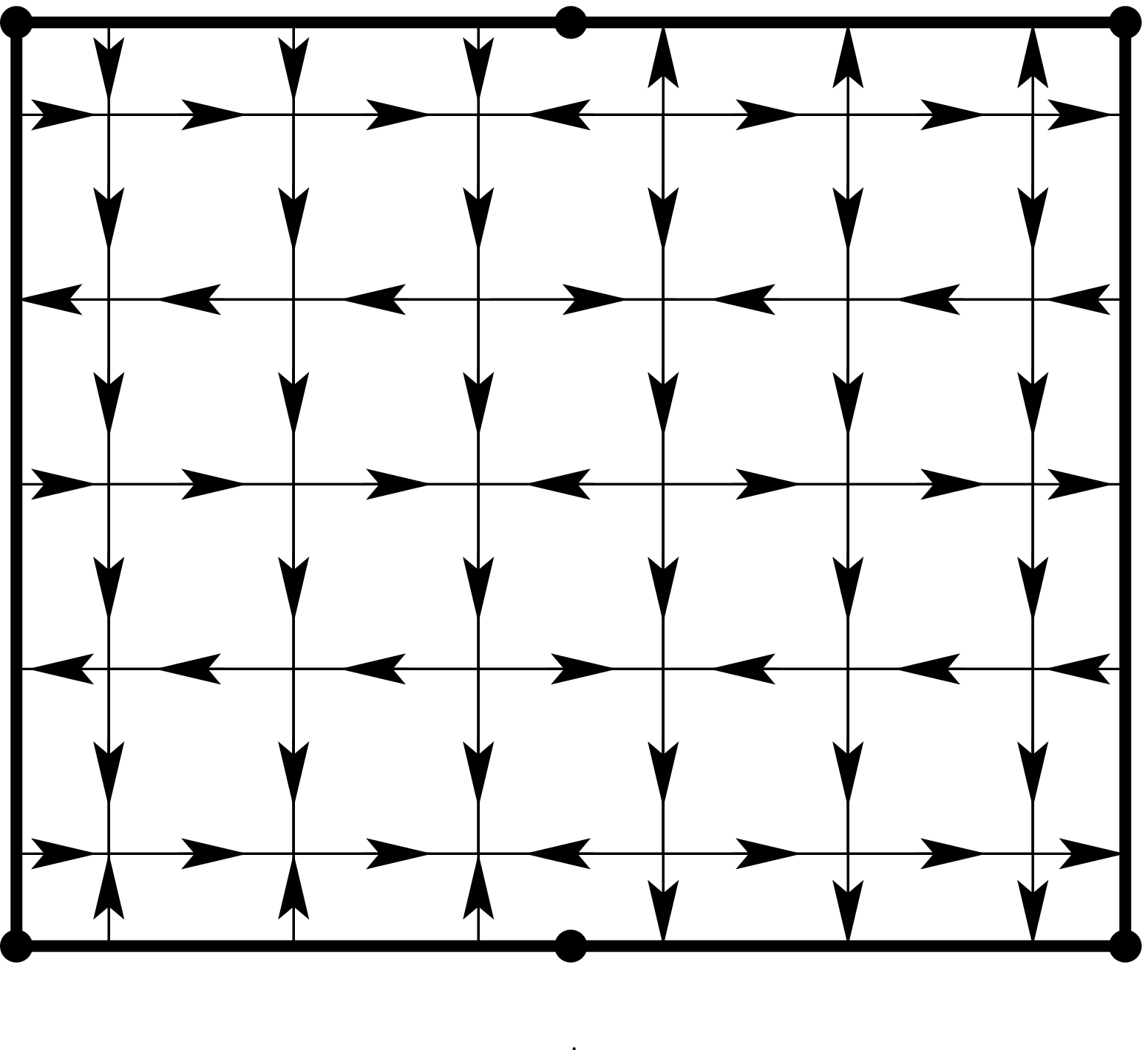,height=4cm}}
\caption{A Kasteleyn orientation on a square lattice in the Klein bottle.}
\label{fig:Klein}
\end{figure}
\end{example}

\subsection{Counting Kasteleyn orientations}

The main result of this section is the following generalization of Proposition~\ref{prop:Kast}.

\begin{theorem}
\label{thm:Kast-n}
There exists a Kasteleyn orientation on $(X,\omega)$ if and only if $X$ has an even number of vertices.
In this case, $\K(X,\omega)$ is an $H^1(\SI;\Z_2)$-torsor.
\end{theorem}

The proof of the first part will rely on the following result.

\begin{lemma}
\label{lemma:parity}
Given any orientation $K$ of $\G$ and any $\omega\in C^1(X;\Z_2)$ representing $w_1$,
the sum $\sum_{f\subset X}c^K(f)$ has the same parity as the number of vertices of $\G$.
\end{lemma}
\begin{proof}
Throughout this demonstration, all integers and equalities are to be considered modulo 2. Let $V$ (resp. $E_0,E_1,F$) denote
the number of vertices (resp. 0-edges, 1-edges, faces) in $(X,\omega)$. Given a face $\widetilde f$ of $\widetilde X$,
let $n_+^K(\partial\widetilde f)$ be the number of clockwise-oriented edges in $\partial\widetilde f$ joining
two vertices labelled $+$. Similarly, let $m_-^K(\partial\widetilde f)$ be the number of counterclockwise-oriented edges in
$\partial\widetilde f$ joining two vertices labelled $-$. Finally, let $m_1^K(f)$ be the number of 1-edges
in $\partial f$ oriented in a fixed direction around $f$. (Since $\partial f$ has an even number of 1-edges, this number is
independent of the choice of this direction.) Fixing a lift $\widetilde f$ of each face $f$ of $X$, we can compute
\[
\sum_{f\subset X}c^K(f)=
\sum_{\widetilde f}(n_+^K(\partial\widetilde f)+m_-^K(\partial\widetilde f))+\sum_{f\subset X}m_1^K(f)+F.
\]
Each 0-edge contributes exactly 1 in the first sum of the right-hand side,
which is therefore equal to $E_0$. It remains to check that the second sum is equal to $E_1+\chi(\SI)$, as it implies
\[
V+\sum_{f\subset X}c^K(\widetilde f)=V+E_0+E_1+\chi(\SI)+F=0.
\]
First note that $S:=E_1+\sum_{f\subset X}m_1^K(f)$ is independent of $K$: indeed, reversing $K$ along an edge
changes the contribution to $S$ of both adjacent faces. Hence, it can be assumed that $K$ is the orientation given by a
global numbering of the vertices of $X$. Furthermore, one easily checks that $S$ remains constant when
an edge $e$ is added that subdivides a face $f$ of $X$ in two. (Note that $\omega(e)$ is determined by $\omega(e')$
for $e'\subset\partial f$.)  Therefore, we can assume that $X$ is a triangulation of $\Sigma$. But for a triangular face
$f$, and with $K$ as above, the cup product $\omega\smile\omega$ satisfies
\[
(\omega\smile\omega)(f)+m_1^K(f)=
\begin{cases}
1 & \text{if $\partial f$ has two 1-edges;} \\
0 & \text{if $\partial f$ has no 1-edge.}
\end{cases}
\]
Summing over all faces, and using the fact that $w_1^2=\chi(\SI)$, we obtain the equality
\[
\chi(\SI)+\sum_{f\subset X}m_1^K(f)=E_1,
\]
which concludes the proof.
\end{proof}

\begin{proof}[Proof of Theorem~\ref{thm:Kast-n}]
Given any orientation $K$ of $\G$, let $c^K\in C^2(X;\Z_2)$ be its Kasteleyn curvature.
$K$ is Kasteleyn if and only if $c^K=0$, in which case the number of vertices $V$ of $\G$ is even by Lemma
\ref{lemma:parity}. Conversely, if $V$ is even, then $\sum_{f\subset X}c^K(f)=0$ by the same lemma. This implies that
$c^K$ is a coboundary, that is, there exists a $\phi\in C^1(X;\Z_2)$ such that $c^K=\delta\phi$.
Consider now the orientation $K^\phi$ which coincides with $K$ on an edge $e$ if and only if $\phi(e)=0$.
Given any face $f$ of $X$, we have the following equality modulo 2:
\[
(\delta\phi)(f)=\phi(\partial f)=\sum_{e\subset\partial f}\phi(e)=c^K(f)+c^{K^\phi}(f).
\]
Since $c^K=\delta\phi$, it follows that $c^{K^\phi}=0$, that is, $K^\phi$ is a Kasteleyn orientation.

Let us now prove the second statement, assuming that $\K(X,\omega)$ is non-empty.
The action of an element $[\phi]\in\coho=H^1(X;\Z_2)$ on $[K]\in\K(X,\omega)$ is defined by $[K]+[\phi]=[K^\phi]$, with
$K^\phi$ as above.
Since $\phi$ is a cocycle, the equation displayed above shows that $K^\phi$ is Kasteleyn if and only if $K$ is.
Note also that $K^\phi$ is equivalent to $K$ if and only if $\phi$ is a coboundary.
Therefore, this action of $\coho$ on $\K(X,\omega)$ is well-defined, and free. Finally, given two
Kasteleyn orientations $K$ and $K'$, let $\phi$ denote the 1-cochain taking value $0$ on an edge $e$ if and only if $K$ and
$K'$ agree on $e$. Obviously, $K'=K^\phi$, and $\phi$ is a cocycle by the identity displayed above. Therefore, the action
is freely transitive.
\end{proof}

\begin{remark}
By the proof of Lemma~\ref{lemma:parity}, the ultimate reason for the existence of a Kasteleyn orientation on $X$ is
the vanishing of the cohomology class $w_1^2+w_2$ in $H^2(\SI;\Z_2)$. This is nothing but the obstruction to the
existence of a $\pin$ structure on the manifold $\SI$.
\end{remark}

The proof of this theorem actually provides us with an algorithm to construct all equivalence classes of Kasteleyn
orientations on a given surface graph $(X,\omega)$ with an even number of vertices.
\begin{ticklist}
\item{Start with any orientation $K_0$ of $\G$, and compute its Kasteleyn curvature $c^{K_0}$, for example, using Remark
\ref{rem:Kast}.}
\item{By Lemma~\ref{lemma:parity}, $c^{K_0}(f)=1$ for an even number of faces. Pick two of them,
join their interior with a curve $\gamma$ in $\SI$ intersecting $\G$ transversally, and invert the orientation
of an edge of $\G$ each time it crosses $\gamma$. The Kasteleyn curvature of the resulting orientation vanishes
at these two faces, and remains unchanged elsewhere. This inductively leads to a Kasteleyn orientation $K$.}
\item{To construct the other Kasteleyn orientations, consider a family of closed curves
$\alpha_1,\dots,\alpha_{b_1}$ intersecting $\G$ transversally, and representing a basis of $\ho$. For any subset
$I\subset \{1,\dots,b_1\}$, let $K^I$ denote the orientation obtained from $K$ by inverting the orientation of an edge
of $\G$ each time it crosses some $\alpha_i$ with $i\in I$. These $K^I$ represent all equivalence classes of Kasteleyn orientations on $(X,\omega)$.}
\end{ticklist}

\subsection{Dependance on the choice of $\omega$}

The definition of a Kasteleyn orientation depends on the choice of the cocycle $\omega$ representing the first
Stiefel-Whitney class of $\SI$. However, two such choices can be naturally related as follows.

\begin{proposition}
\label{prop:omega}
Given $\omega,\omega'$ representing the first Stiefel-Whitney class $w_1$, there is an $\coho$-equivariant bijection
\[
\varphi_{\omega'\omega}\colon\K(X,\omega)\longrightarrow\K(X,\omega')
\]
which satisfies the relations $\varphi_{\omega\omega}=\mathit{id}$ and
$\varphi_{\omega''\omega'}\circ \varphi_{\omega'\omega}=\varphi_{\omega''\omega}$.
\end{proposition}
\begin{proof}
Let $\omega,\omega'\in C^1(X;\Z_2)$ be two representatives of $w_1$. Since $\omega$ and $\omega'$ are cohomologous,
they can be obtained from one another by flipping all $0$'s and $1$'s around the vertices in some set $S$. To prove the
proposition, it is enough to check that flipping around one vertex $v$ induces an $\coho$-equivariant map
$\varphi_v\colon\K(X,\omega)\to\K(X,\omega+\delta v)$
such that $\varphi_v\circ\varphi_v=\mathit{id}$ and $\varphi_v\circ\varphi_{v'}=\varphi_{v'}\circ\varphi_v$
for any two vertices $v,v'$. Indeed, we can then define $\varphi_{\omega'\omega}$ as the composition (in any order)
of all the $\varphi_v$'s with  $v\in S$.

Let $\varphi_v$ be defined by $\varphi_v([K])=[K']$, where $K'$ agrees with $K$ on an edge $e$ unless $v\in\partial e$
and $\omega(e)=1$. It is easy but tedious to check that if $K$ is Kasteleyn on $(X,\omega)$, then $K'$ is Kasteleyn on
$(X,\omega+\delta v)$. On the other hand, it is then obvious that $\varphi_v$ is a well-defined equivariant map. The identity
$\varphi_v\circ\varphi_{v'}=\varphi_{v'}\circ\varphi_v$ is also immediate. Finally, $\varphi_v\circ\varphi_v$ maps $[K]$ to
the class of $K''$, the orientation obtained from $K$ by flipping the orientations of all the edges adjacent to $v$.
Hence, $K''$ and $K$ are equivalent, so $\varphi_v\circ\varphi_v$ is the identity.
\end{proof}

\section{Kasteleyn orientations as discrete pin$^-$ structures}
\label{sec:pin}

\subsection{Basic facts about pin$^-$ structures}

We shall now informally review several general facts about $\pin$ structures, which are the natural generalization of
spin structures to non-orientable manifolds. We refer to \cite{K-T} for details and proofs.

Recall that $\mathit{Pin}^-(n)$ is a topological group which is a double cover of the orthogonal group $O(n)$. A
{\em $\mathit{pin}^-$ structure}
on an n-dimensional Riemannian manifold $M$ is a $\pin$ structure on its frame bundle $P_O\to M$,
that is, a principal $\mathit{Pin}^-(n)$-bundle $P\to M$ together with a 2-fold covering map $P\to P_O$ which restricts to
$\mathit{Pin}^-(n)\to O(n)$ on each fiber. The obstruction to putting a $\pin$ structure on $M$ is
$w_2+w_1^2\in H^2(M;\Z_2)$. If this class vanishes, then the set $\mathrm{Pin}^{-}(M)$ of equivalence classes of
$\pin$ structures on $M$ is an $H^1(M;\Z_2)$-torsor.
The following special case of \cite[Lemma 1.7]{K-T} will be essential for our purpose: there is an $H^1(M;\Z_2)$-equivariant
bijection between $\mathrm{Pin}^{-}(M)$ and the set of equivalence classes of spin structures on $\xi\oplus\det\xi$, where
$\xi$ denotes the tangent bundle of $M$ and $\det\xi$ the determinant line bundle. (Note that $\det\xi$ is simply
the line bundle corresponding to the orientation cover $\widetilde{M}\to M$ viewed as a principal $O(1)$-bundle.)

The 2-dimensional case is particularly easy to deal with.
First of all, any compact surface $\SI$ admits a $\pin$ structure, as $w_2$ and $w_1^2$ are both equal to the Euler
characteristic of $\SI$ modulo 2. Hence, the set $\mathrm{Pin}^-(\SI)$ is an $\coho$-torsor.
Furthermore, a spin structure on $\xi\oplus\det\xi$ is nothing but a trivialisation of this bundle. Let
$\lambda\colon E(\lambda)\to\SI$ denote the determinant line bundle, and let $p\colon TE(\lambda)\to E(\lambda)$ be
the tangent bundle of its total space. By the following commutative diagram of bundles,
\[
\xymatrix{
TE(\lambda) \ar[r] \ar[d]_p & E(\xi\oplus\lambda) \ar[d]^{\xi\oplus\lambda}\\
E(\lambda) \ar[r]^\lambda & \SI
}
\]
$\xi\oplus\lambda$ is the restriction of the tangent bundle of $E(\lambda)$ to $\SI$. (Here, $\SI$ embeds in
$E(\lambda)$ as the 0-section of $\lambda$.) Therefore, a $\pin$ structure on a surface $\SI$ is a trivialisation
over $\SI$ of the vector bundle $TE(\lambda)\to E(\lambda)$.

Finally, Johnson's theorem \cite{Jo} generalizes to non-orientable surfaces as follows. (Again, we refer to \cite{K-T}
for a proof.) A function $q\colon\ho\to\Z_4$
is called a {\em quadratic enhancement\/} of the intersection form if $q(x+y)=q(x)+q(y)+2(x\cdot y)$ for
all $x,y\in\ho$, where $\cdot$ denotes the intersection form, and $2\colon\Z_2\to\Z_4$ the inclusion homomorphism.
One easily checks that the set $\mathrm{Quad}(\SI)$ of such quadratic enhancements admits a freely transitive action of
$\coho=\Hom(\ho;\Z_2)$ given by $\phi\ast q=q+2\phi$. The statement generalizing Johnson's theorem is the following:
There is an $\coho$-equivariant bijection $\mathrm{Pin}^-(\SI)=\mathrm{Quad}(\SI)$.

More explicitely, consider an element of $\mathrm{Pin}^-(\SI)$, that is, a trivialisation of the tangent bundle
$p\colon TE(\lambda)\to E(\lambda)$ over $\SI$. Then, the corresponding quadratic enhancement $q\colon\ho\to\Z_4$
is determined by its value on the class of an embedded circle $C$ in $\SI$; this value is obtained as follows.
Let $\tau$ denote the restriction of $p$ to $p^{-1}(C)$. Obviously, the $\pin$ structure induces a trivialisation $s$ of
$\tau$. Note also that $\tau=TC\oplus\nu(C\subset\SI)\oplus\nu(\SI\subset E(\lambda))$, where $\nu$ denotes the normal
bundle. Pick $x\in C$, and orient these three line bundles at $x$ so that the induced orientation on $\tau$ agrees with
the one given by $s$. Now, the orientation of $TC$ determines a trivialisation $\sigma$ of this line bundle. Pick a framing
$s'$ of $\nu(C\subset\SI)\oplus\nu(\SI\subset E(\lambda))$ such that $\sigma\oplus s'$ is homotopic to $s$. Then, $q([C])$
is given by the class modulo 4 of $h_{s'}(C)+2$, where $h_{s'}(C)$ denotes the number of right half twists that
$\nu(C\subset\SI)$ makes with respect to $s'$ in a complete traverse of $C$.

\subsection{Encoding a pin$^-$ structure}

Let us try to encode combinatorially a $\pin$ structure on a surface $\SI$, that is, a trivialisation of
$TE(\lambda)\to E(\lambda)$ over $\SI$. First note that, if $\varphi\colon\widetilde\SI\to\widetilde\SI$ denotes the
involution of the orientation cover of $\SI$, then $E(\lambda)$ can be expressed as the quotient of $\widetilde\SI\times\R$
by the action of $\varphi\times -\id$. Therefore, a $\pin$ structure on $\SI$ is equivalent to a trivialisation of
$T\widetilde\SI\times\R\to\widetilde\SI$, invariant under the action of $d\varphi\times-\id$. This is what we will encode.

Fix a cellular decomposition $X$ of $\SI$ and a representative $\omega\in C^1(X;\Z_2)$ of $w_1$. As mentioned
in Section~\ref{sec:Kast-n}, $\omega$ determines a labelling of the vertices of $\widetilde X$ with signs,
such that $\omega(e)=0$ if and only if the two endpoints of a lift of $e$ have the same label.
This in turn induces an orientation on $\widetilde X$.

$\bullet$
To construct a framing of $T\widetilde\SI\times\R\to\widetilde\SI$ over $\widetilde X^0$, fix a dimer configuration $D$
on $X^1$. This lifts to a dimer configuration $\widetilde D$ on $\widetilde X^1$ which determines a (unit length)
vector field $s_1$ along vertices of $\widetilde{X}$. It can be completed by a vector field $s_2$, so that
$(s_1(x_+),s_2(x_+))$ is a positive orthonormal basis of $T_{x_+}\widetilde\SI$ and
$(s_1(x_-),s_2(x_-))$ is a negative orthonormal basis of $T_{x_-}\widetilde\SI$. Setting $s_3(x_{\pm})=\pm 1$
gives a framing $s=(s_1,s_2,s_3)$ of $T\widetilde\SI\times\R\to\widetilde\SI$ over $\widetilde X^0$,
which is clearly invariant under the action of $d\varphi\times -\id$.

\begin{figure}[htbp]
\labellist\small\hair 2.5pt
\pinlabel {$K$} at 315 60
\pinlabel {$D$} at 120 35
\pinlabel {$D$} at 535 130
\pinlabel {$x$} at 135 100
\pinlabel {$y$} at 525 80
\pinlabel {$\SI$} at 100 140
\pinlabel {$\widetilde{\SI}$} at 100 370
\pinlabel {$x_+$} at 135 325
\pinlabel {$y_+$} at 525 300
\pinlabel {$s_1$} at 110 270
\pinlabel {$s_2$} at 193 292
\pinlabel {$s_3$} at 162 395
\endlabellist
\centerline{\psfig{file=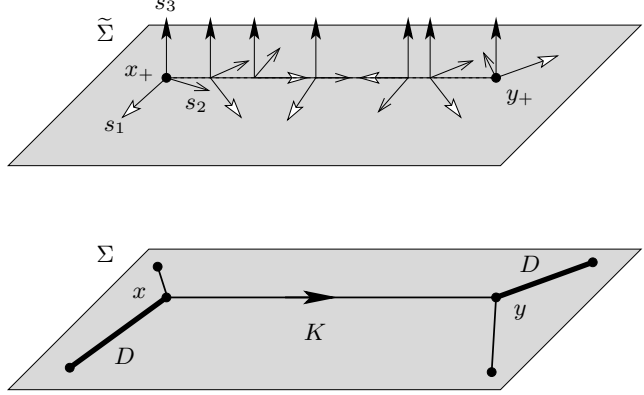,height=5cm}}
\caption{Extension of the framing along a 0-edge.}
\label{fig:0-edge}
\end{figure}

$\bullet$
To extend this framing to the 1-skeleton of $\widetilde{X}$, fix an orientation $K$ of the edges of $X^1$. If $e$ denotes
a 0-edge oriented from a vertex $x$ to a vertex $y$, then the framing along the lift $\widetilde{e}$ between $x_\pm$ and
$y_\pm$ is defined as follows. Moving along $\widetilde{e}$, first make a right-hand rotation of the framing around the axis
$s_3$ until $s_1$ points in the direction of the edge; then make a left-handed half twist around $s_3$ so that
$s_1$ points backward; finally, make a right-hand rotation around $s_3$ until $s_1$ coincides with
$s_1(y_\pm)$. This construction is illustrated in Figure~\ref{fig:0-edge}.

If $e$ denotes a 1-edge oriented from $x$ to $y$, then the framing along $\widetilde{e}$ between $x_\pm$ and
$y_\mp$ is defined as follows. First make a right-hand rotation around the axis $s_3$ until $s_1$ points in
the direction of the edge; then make a right-handed half twist around $s_2$; finally, make a left-hand rotation
around $s_3$ until $s_1$ coincides with $s_1(y_\mp)$. This is illustrated in Figure~\ref{fig:1-edge}.

\begin{figure}[htbp]
\labellist\small\hair 2.5pt
\pinlabel {$K$} at 380 40
\pinlabel {$D$} at 170 95
\pinlabel {$D$} at 590 36
\pinlabel {$x$} at 127 62
\pinlabel {$y$} at 634 62
\pinlabel {$\SI$} at 80 115
\pinlabel {$\widetilde{\SI}$} at 80 305
\pinlabel {$x_+$} at 135 230
\pinlabel {$y_-$} at 625 260
\pinlabel {$s_1$} at 125 295
\pinlabel {$s_2$} at 70 225
\pinlabel {$s_3$} at 150 335
\endlabellist
\centerline{\psfig{file=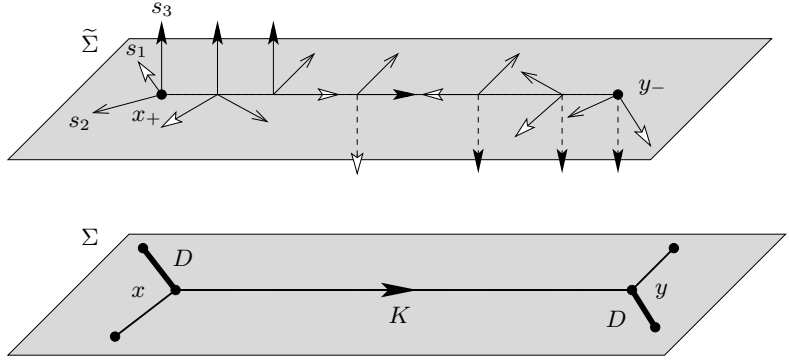,height=4.5cm}}
\caption{Extension of the framing along a 1-edge.}
\label{fig:1-edge}
\end{figure}

$\bullet$
Given a cocycle $\omega$, a dimer configuration $D$ on $X^1$ and an orientation $K$ of the edges of $X^1$, we now have
a well-defined framing $s=s(\omega,D,K)$ of $T\widetilde\SI\times\R\to\widetilde\SI$ over $\widetilde X^1$,
invariant under the action of $d\varphi\times -\id$, which we wish to extend to the whole of $\widetilde X$. 
Let $\widetilde f$ be a face of $\widetilde{X}$, and let us
fix a constant framing of $T\widetilde{\SI}\times\R$ over $\widetilde f$. Then, the restriction of $s$ to
$\partial\widetilde f$ defines a loop $s(\partial\widetilde f)$ in $SO(3)$.
The framing $s$ extends to $\widetilde f$ if and only if the homotopy class $[s(\partial\widetilde f)]$
is trivial in $\pi_1(SO(3))=\Z_2$. We shall simply denote by $[\widetilde{f}]$ this class in $\Z_2$.

\begin{proposition}
\label{prop:curvature}
Given any face $f$ of $X$ and any lift $\widetilde{f}$ of $f$ in $\widetilde{X}$, $[\widetilde f]$ is equal
to the Kasteleyn curvature $c^K(f)\in\Z_2$. Hence, $s$ extends to $\widetilde{X}$ if and only if $K$ is a Kasteleyn
orientation.
\end{proposition}

\begin{proof}
Given a face $f$ of $X$, recall that
\[
c^K(f)=
n^{\widetilde{K}}(\partial\widetilde f)+
\#\{\text{edges in $\partial\widetilde f$ joining two vertices labelled $-$}\}+1\in\Z_2,
\]
where $\widetilde f$ is any lift of $f$.
First observe that $[\widetilde{f}]$ changes when $K$ is inverted on one edge of $\partial f$. Since the same obviously
holds for $c^K(f)$, it may be assumed that all edges in $\partial\widetilde f$ are oriented counterclockwise, except those joining two vertices labelled with $-$. (In this case, we shall say that $\partial\widetilde f$ is well-oriented.)
It remains to check that $[\widetilde{f}]=1$ whenever $\partial f$ is well-oriented.
Let us prove this by induction on $n$, the number of edges in $\partial\widetilde f$.
A face $f$ with $n=2$ boundary edges is well-oriented if and only if these edges do not have the same orientation.
If they do have the same orientation, then the framing obviously extends to the whole of $\widetilde f$, so that
$[\widetilde{f}]=0$. By the observation above, it follows that $[\widetilde{f}]=1$ if $\partial f$ is well-oriented.
The case $n=3$ can be checked by direct inspection. Consider now
a face $\widetilde f$ with $n\ge 4$ boundary edges. Using one more time the observation above, we have
\[
[\widetilde{f}]=
\left[\begin{array}{c}\includegraphics[height=0.3in]{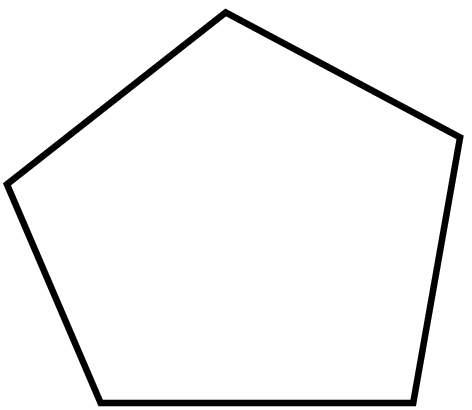}\end{array}\right]=
\left[\begin{array}{c}\includegraphics[height=0.3in]{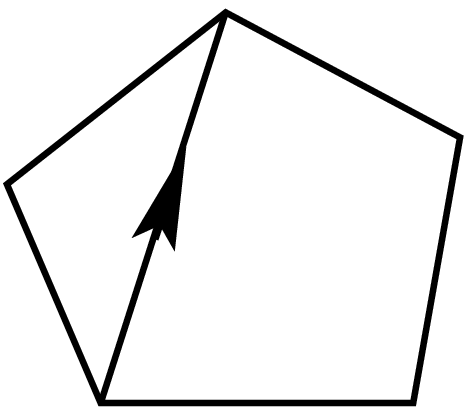}\end{array}\right]=
\left[\begin{array}{c}\includegraphics[height=0.3in]{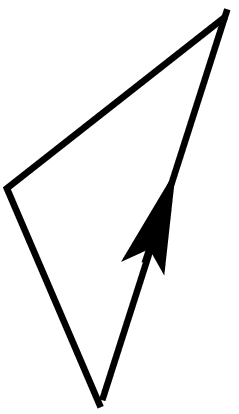}\end{array}\right]+
\left[\begin{array}{c}\includegraphics[height=0.3in]{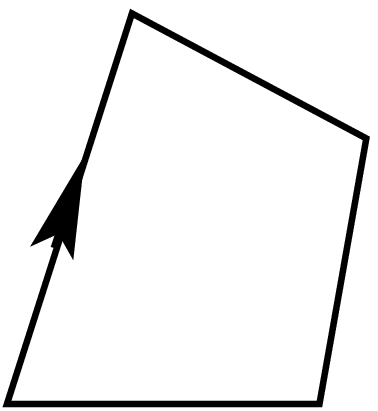}\end{array}\right]=
\left[\begin{array}{c}\includegraphics[height=0.3in]{c8.eps}\end{array}\right]+
\left[\begin{array}{c}\includegraphics[height=0.3in]{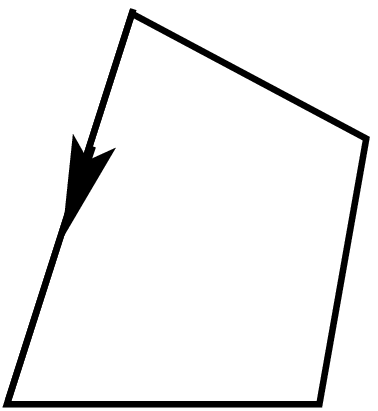}\end{array}\right]+1,
\]
where the orientation of the new edge is chosen so that the resulting triangle is well-oriented. By induction,
$[\widetilde{f}]=1+1+1=1$, and we are done.
\end{proof}

Therefore, a dimer configuration $D$ on $X^1$ and a Kasteleyn orientation $K$ on $(X,\omega)$ determine a $\pin$ structure
on $\SI$, that is, a quadratic enhancement that we shall denote by $q_D^{K,\omega}\colon\ho\to\Z_4$.
It is characterized by the following property.

\begin{proposition}
\label{prop:quad-en}
Let $C$ be an oriented simple closed curve on $X^1$. Then,
\[
q_D^{K,\omega}([C])=2(n^K(C)+\ell^\omega_D(C)+1)+\omega(D\cap C)-\omega(C\setminus D)\pmod{4},
\]
where $\ell^\omega_D(C)$ denotes the number of vertices $x$ in $C$ such that the following condition holds:
$(C,D)$ induces a local orientation at $x\in\SI$ which lifts to the positive orientation of $\widetilde{\SI}$ at $x_+$
(or equivalently, to the negative one at $x_-$).
\end{proposition}
\begin{proof}
Let $s=s(\omega,D,K)$ denote the framing of $T\widetilde\SI\times\R\to\widetilde\SI$ over a lift $\widetilde{C}$ of
$C$, as constructed above. Let $s'$ be a trivialisation of
$\nu(\widetilde{C}\subset\widetilde{\SI})\oplus\nu(\widetilde{\SI}\subset\widetilde{\SI}\times\R)$ such that $s$ is homotopic
to $\sigma\oplus s'$, where $\sigma$ denotes the trivialisation of $T\widetilde{C}$ given by the orientation of
$\widetilde{C}$. Then, $q_D^{K,\omega}([C])$ is equal to the class modulo 4 of $h_{s'}(\widetilde{C})+2$, where
$h_{s'}(\widetilde{C})$ denotes the number of right half twists that $\nu(\widetilde C\subset\widetilde\SI)$
makes with respect to $s'$ in a complete traverse of $\widetilde C$. Hence, we are left with the proof of the equality
\[
h_{s'}(\widetilde{C})\equiv 2(n^K(C)+\ell^\omega_D(C))+\omega(D\cap C)-\omega(C\setminus D)\pmod{4}.\tag{$\ast$}
\]
Note that this equation makes sense for any orientation $K$ of the edges of $X^1$, not only for Kasteleyn orientations;
we shall prove it for every orientation.
Note also that if $s$ and $s_0$ denote two framings over $\widetilde{C}$, then they define a loop in $SO(3)$, and the
difference $h_{s_0'}(\widetilde{C})-h_{s'}(\widetilde{C})$ is equal to twice the class of this loop in $\pi_1(SO(3))=\Z_2$.
In particular, if $K_0$ is obtained by reversing the orientation $K$ on one edge of $C$, and $s,s_0$ denote the induced
framings over $\widetilde{C}$, then $h_{s_0'}(\widetilde{C})-h_{s'}(\widetilde{C})=2$ by the observation at the beginning
of the proof of Proposition~\ref{prop:curvature}.
Hence, $h_{s'}(\widetilde C)$ changes by $\pm 2$ when $K$ is inverted along one edge of $C$.
Since the same obviously holds for the right-hand side of $(\ast)$, it may be assumed that $n^K(C)=0$.
Furthermore, $h_{s'}(\widetilde C)$ also changes by $\pm 2$ when a dimer of $\widetilde{D}$ pointing out to the left of
$\widetilde C$ at a vertex $x_+$ (resp. to the right of $\widetilde C$ at a vertex $x_-$) is replaced by a dimer pointing
out to the right of $\widetilde C$ (resp. to the left). This follows from the following computation, which makes use of
Proposition~\ref{prop:curvature}:
\[
h_{s'}\left(\begin{array}{c}\includegraphics[width=0.5in]{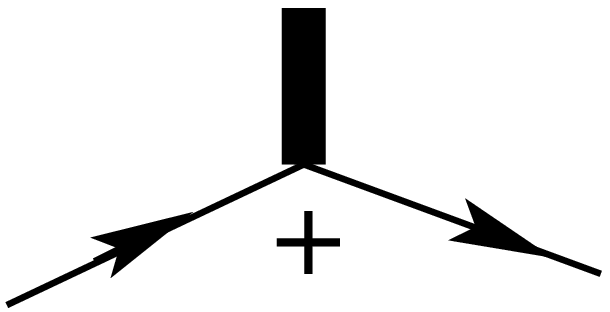}\end{array}\right)-
h_{s_0'}\left(\begin{array}{c}\includegraphics[width=0.5in]{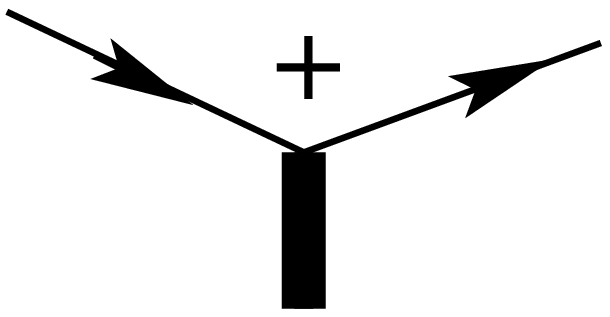}\end{array}\right)=
2\left[\begin{array}{c}\includegraphics[width=0.5in]{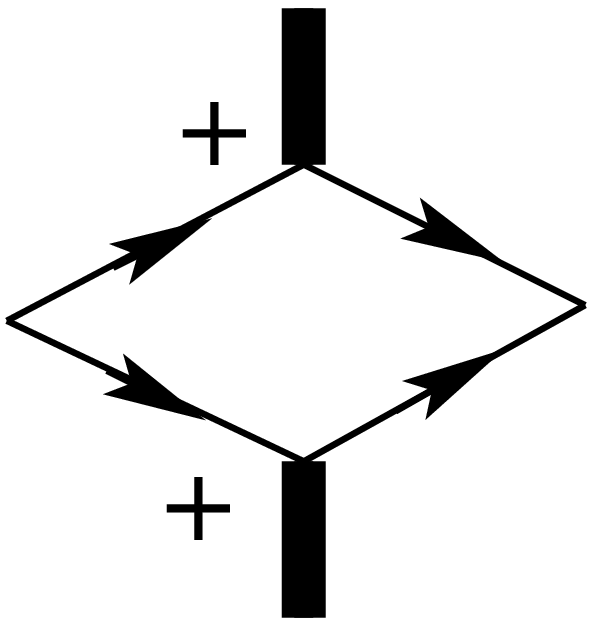}\end{array}\right]=
2\, c^K\left(\begin{array}{c}\includegraphics[width=0.5in]{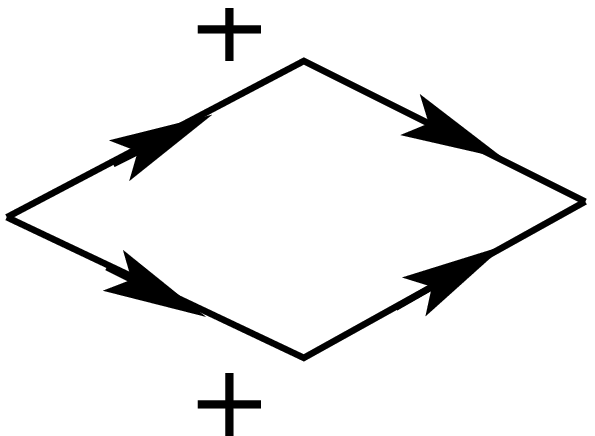}\end{array}\right)=2.
\]
Since the same holds for the right-hand side of $(\ast)$,
it may be assumed that $\ell^\omega_D(C)=0$. Similar arguments allow us to
assume that no dimer of $D$ lies in $C$, so that $\omega(D\cap C)=0$ and $\omega(C\setminus D)=\omega(C)$.
Hence, it may be assumed that $K$ agrees everywhere with the orientation on $C$, and that the dimer
of $D$ adjacent to a vertex $x\in C$ always lies outside $C$, so that the local orientation at $x\in C$ induced by $(C,D)$
lifts to the negative orientation at $x_+$ (and the positive one at $x_-$). But in this case, the framing constructed
in Figure~\ref{fig:0-edge} along the 0-edges is homotopic to a constant framing.
Also, the framing constructed in Figure~\ref{fig:1-edge} along the 1-edges is homotopic to a framing of the form
$\sigma\oplus s'$, with $s'$ making one right half twist along each 1-edge.
Therefore, $h_{s'}(\widetilde{C})=-\omega(C)$, and the proposition is proved.
\end{proof}

\subsection{The correspondence theorem}

We can now state our correspondence theorem, which generalizes Theorem~\ref{thm:corr} to the (possibly) non-orientable case.

\begin{theorem}
\label{thm:n-corr}
Let $X$ be a cellular decomposition of a closed surface $\SI$, and let $\omega\in C^1(X;\Z_2)$ be a representative of the
first Stiefel-Whitney class $w_1$. Then, any dimer configuration $D\in\D(X^1)$ induces an
$H^1(\SI;\Z_2)$-equivariant bijection
\[
\psi^\omega_D\colon\K(X,\omega)\to\mathrm{Quad}(\SI)=\mathrm{Pin}^-(\SI),\quad [K]\mapsto q_D^{K,\omega}
\]
from the set of equivalence classes of Kasteleyn orientations on $(X,\omega)$ to the set of equivalence classes of
$\mathit{pin}^-$ structures on $\SI$. Furthermore, given another dimer configuration $D'\in\D(X^1)$, $\psi^\omega_{D'}$ is obtained from $\psi^\omega_D$ by action of the Poincar\'e dual to $[D+D']\in H_1(\SI;\Z_2)$.
Finally, given another representative $\omega'$ of $w_1$, the following  diagram is commutative,
\[
\xymatrix{
\K(X,\omega)
\ar[rr]^{\varphi_{\omega'\omega}}
\ar[dr]_{\psi^\omega_D}
&& \K(X,\omega')
\ar[dl]^{\psi^{\omega'}_D} \\
& \mathrm{Pin}^-(\SI)} \\
\]
where $\varphi_{\omega'\omega}$ is the equivariant bijection defined in Proposition~\ref{prop:omega}.
\end{theorem}
\begin{proof}
Given fixed $\omega\in C^1(X;\Z_2)$ and $D\in\D(X^1)$, the construction above associates to each Kasteleyn orientation
$K$ on $(X,\omega)$ a quadratic enhancement $q^{K,\omega}_D\in\mathrm{Quad}(\SI)$ which is determined by the equality in
Proposition~\ref{prop:quad-en}. Therefore, we only need to check our statements for the evaluation of these quadratic
enhancements on the homology class of simple closed curves in $X^1$.

If $K'$ is equivalent to $K$, then $n^{K'}(C)=n^K(C)$ for any oriented simple closed
curve $C$ in $X^1$. Therefore, $q^{K',\omega}_D$ is equal to $q^{K,\omega}_D$ and we have a well-defined map
$\psi^\omega_D\colon\K(X,\omega)\to\mathrm{Quad}(\SI)$. Furthermore, if $K^\phi$ is the orientation obtained from $K$ by
action of a cocycle $\phi\in C^1(X;\Z_2)$ -- recall the proof of Theorem~\ref{thm:Kast-n} --
then $n^{K^\phi}(C)=n^K(C)+\phi(C)$. Hence,  $q^{K^\phi,\omega}_D=q^{K,\omega}_D+2[\phi]$,
so $\psi^\omega_D$ is $\coho$-equivariant.

Now, let $D,D'$ be two dimer configurations on $X^1$, and consider the associated quadratic enhancements $q^{K,\omega}_D$
and $q^{K,\omega}_{D'}$. For an oriented simple closed curve $C$ in $X^1$, Proposition~\ref{prop:quad-en} leads to
the equality
\[
(q^{K,\omega}_D-q^{K,\omega}_{D'})([C])\equiv 2(\ell^\omega_{D'}(C)+\ell^\omega_D(C)+\omega(C\cap(D+D')))\pmod{4},
\]
where $D+D'$ denotes the disjoint simple cycles obtained by adding $D,D'\in C_1(X;\Z_2)$. Using the definition of
$\ell^\omega_D(C)$, one checks that 
\[
\ell^\omega_{D'}(C)+\ell^\omega_D(C)+\omega(C\cap(D+D'))\equiv C\cdot(D+D')\pmod{2}.
\]
Therefore, $q^{K,\omega}_{D'}=q^{K,\omega}_D+2 [D+D']^\ast$, showing the second claim.

Let us finally prove the equality $\psi^{\omega'}_D\circ\varphi_{\omega'\omega}=\psi^\omega_D$.
By construction of $\varphi_{\omega'\omega}$, one only needs to check the following: given a Kasteleyn orientation $K$
on $(X,\omega)$, an oriented simple closed curve $C$ in $X^1$, and a vertex $x\in C$, we have the equality
$q^{K,\omega}_D([C])=q^{K',\omega'}_D([C])$, where $\omega'$ is obtained from $\omega$ by changing the $0$'s and $1$'s
labelling all edges adjacent to $x$, and $K'$ is obtained from $K$ by inverting the orientation of all the edges $e$ adjacent
to $x$ such that $\omega(e)=1$. By Proposition~\ref{prop:quad-en}, the difference
$\Delta=q^{K,\omega}_D([C])-q^{K',\omega'}_D([C])\in\Z_4$ is given by
\[
2(n^K(C)+n^{K'}(C))+2(\ell^\omega_D(C)+\ell^{\omega'}_D(C))+(\omega-\omega')(D\cap C)+(\omega'-\omega)(C\setminus D).
\]
The first term above is equal to $2(\omega(e_1)+\omega(e_2))$, where $e_1,e_2$ denote the two edges of $C$ adjacent to $x$.
If neither $e_1$ nor $e_2$ is a dimer of $D$, then
\begin{eqnarray*}
\Delta &=& 2(\omega(e_1)+\omega(e_2))+2+\omega'(e_1)+\omega'(e_2)-\omega(e_1)-\omega(e_2)\\
	&=&(\omega(e_1)+\omega'(e_1)+1)+(\omega(e_2)+\omega'(e_2)+1).
\end{eqnarray*}
Each of these terms is equal to $2$, so the sum is zero modulo 4.
On the other hand, if one of these edges (say, $e_1$) is a dimer of $D$, then
\begin{eqnarray*}
\Delta &=& 2(\omega(e_1)+\omega(e_2))+\omega(e_1)-\omega'(e_1)+\omega'(e_2)-\omega(e_2)\\
	&\equiv& (\omega(e_2)+\omega'(e_2))-(\omega(e_1)+\omega'(e_1))\pmod{4}.
\end{eqnarray*}
Each of these terms is equal to $1$, so the difference is zero.
\end{proof}

As stated above, the correspondence depends on the choice of $D\in\D(X^1)$. This can be remedied as follows. Let
$\mathcal{B}=\{\alpha_i\}$ denote a family of closed curves in $\SI$, transverse to $X^1$, whose classes
form a basis of $H_1(\SI;\Z_2)$. Given any $D\in\D(X^1)$, let $\varphi_\mathcal{B}^D\in\coho=\Hom(\ho,\Z_2)$ be given by
$\varphi_\mathcal{B}^D([\alpha_i])=\alpha_i\cdot D$. Finally, define $q^{K,\omega}_\mathcal{B}\in\mathrm{Quad}(\SI)$
by the equality $q^{K,\omega}_\mathcal{B}=q^{K,\omega}_D+2\varphi_\mathcal{B}^D$.

\begin{corollary}
\label{cor:n-corr}
Let $X$ be a cellular decomposition of a closed surface $\SI$ such that $X^1$ admits a dimer configuration
$D$, and let $\omega\in C^1(X;\Z_2)$ be a representative of the first Stiefel-Whitney class $w_1$. Then, the map
\[
\psi^\omega_\mathcal{B}\colon\K(X,\omega)\to\mathrm{Quad}(\SI)=\mathrm{Pin}^-(\SI),\quad [K]\mapsto q_\mathcal{B}^{K,\omega}
\]
is an $H^1(\SI;\Z_2)$-equivariant bijection which does not depend on $D$. Furthermore, given another representative
$\omega'$ of $w_1$, we have the equality $\psi^{\omega'}_\mathcal{B}\circ\varphi_{\omega'\omega}=\psi^\omega_\mathcal{B}$.
\end{corollary}
\begin{proof}
The demonstration of Corollary~\ref{cor:corr} extends verbatim.
\end{proof}

\section{The Pfaffian formula in the non-orientable case}
\label{sec:Pf-n}

We shall now use the previous section to derive the Pfaffian formula in the general case of a graph embedded in
a possibly non-orientable surface.

Let $\G$ be a finite connected graph endowed with an edge weight system $\mathrm{w}$.
If $\G$ does not admit any dimer configuration, then the partition function vanishes.
Let us therefore assume that $\G$ admits a dimer
configuration $D_0$. Enumerate the vertices of $\G$ by $1,2,\dots,2n$ and embed $\G$ in a closed
surface $\SI$ as the 1-skeleton of a cellular decomposition $X$ of $\SI$. Finally, let us fix an $\omega\in C^1(X;\Z_2)$ 
representing the first Stiefel-Whitney class of $\SI$.

Since $\G$ has an even number of vertices, Theorem~\ref{thm:Kast-n} ensures that the set $\K(X,\omega)$ is an
$H^1(\SI;\Z_2)$-torsor. In particular, there exists a Kasteleyn orientation $K$ on $(X,\omega)$.
We define an associated skew-symmetric matrix $A^{K,\omega}$ as follows: its coefficients are given by
\[
a_{jk}=\sum_{e}\e_{jk}^K(e)i^{\omega(e)}\mathrm{w}(e),
\]
where the sum is on all edges $e$ in $\Gamma$ between the vertices $j$ and $k$, and
$\e^K_{jk}(e)=+1$ (resp. $-1$) if $e$ is oriented by $K$ from $j$ to $k$ (resp. from $k$ to $j$).
In short, it is exactly the matrix defined by Kasteleyn, but with all weights of the 1-edges multiplied by $i=\sqrt{-1}$.
Given a dimer configuration $D$, we shall simply denote by $\omega(D)$ the sum $\sum_{e\subset D}\omega(e)$. Recall also
the notation $\e^K(D)=\pm 1$ introduced in (\ref{equ:eps}).

Finally, recall that any quadratic enhancement $q\colon V\to\Z_4$
of a non-singular bilinear form on a $\Z_2$-vector space $V$ has a well-defined
{\em Brown invariant\/} $\beta(q)\in\Z_8$ (see e.g.~\cite{K-T}). It is given by the equality
\[
\exp(i\pi/4)^{\beta(q)}=\frac{1}{\sqrt{|V|}}\sum_{x\in V}i^{q(x)}.
\]

\begin{theorem}
\label{thm:Pf-n}
The partition function of the dimer model on $\Gamma$ is given by the formula
\[
Z=\frac{(-i)^{\omega(D_0)}}{2^{b_1/2}}\sum_{[K]\in\K(X,\omega)}
\exp(i\pi/4)^{\beta(q^{K,\omega}_{D_0})}\e^K(D_0)\Pf(A^{K,\omega}),
\]
where $b_1=\dim H_1(\SI;\Z_2)$ and $\beta(q)$ denotes the Brown invariant of the quadratic enhancement $q$.
Furthermore, each term of this sum depends only on the class of $K$ in $\K(X,\omega)$, but neither on the choice
of a representative of this class, nor on the choice of $D_0$.
\end{theorem}
\begin{proof}
Given any Kasteleyn orientation $K$ of $\G$, we have
\begin{align*}
\e^K(D_0)i^{-\omega(D_0)}\Pf(A^{K,\omega})&\overset{(\ref{equ:Pf})}{=}\sum_{D\in\D(\G)}\e^K(D_0)\e^K(D)
\,i^{\omega(D)-\omega(D_0)}\,\mathrm{w}(D)\\
	&\overset{(\ref{equ:cc})}{=}\sum_{D\in\D(\G)}(-1)^{\sum_j (n^K(C_j)+1)}\,i^{\omega(D)-\omega(D_0)}\,\mathrm{w}(D)\\
	&=\sum_{D\in\D(\G)}i^{\,\sum_j (2n^K(C_j)+2+\omega(C_j\setminus D_0)-\omega(C_j\cap D_0))}\,\mathrm{w}(D),
\end{align*}
where the $C_j$'s are the disjoint cycles forming $D+D_0$.
At any vertex of $C_j$, the adjacent dimer of $D_0$ lies on $C_j$, so $\ell^\omega_{D_0}(C_j)=0$.
Since the cycles $C_j$ are disjoint, Proposition~\ref{prop:quad-en} gives
\[
\sum_j(2n^K(C_j)+2+\omega(C_j\setminus D_0)-\omega(C_j\cap D_0))=-\sum_jq_{D_0}^{K,\omega}([C_j])
=-q^{K,\omega}_{D_0}([D+D_0]).
\]
Therefore, every element $[K]$ of $\K(X,\omega)$ induces a linear equation
\[
\e^K(D_0)i^{-\omega(D_0)}\Pf(A^{K,\omega})=\sum_{\alpha\in H_1(\SI;\Z_2)}i^{-q^{K,\omega}_{D_0}(\alpha)}Z_\alpha(D_0),
\]
where $Z_\alpha(D_0)=\sum_{[D+D_0]=\alpha}\mathrm{w}(D)$, the sum being over all $D\in\D(\G)$ such that $[D+D_0]=\alpha$.
One can solve this linear system of $2^{b_1}$ equations with $2^{b_1}$ unknowns as in \cite[Theorem 5]{C-RI}, obtaining
\[
Z_\alpha(D_0)=\frac{1}{2^{b_1}}\sum_{[K]}i^{q^{K,\omega}_{D_0}(\alpha)}\e^{K}(D_0)i^{-\omega(D_0)}\Pf(A^{K,\omega}).
\]
The final formula for $Z$ is now obtained by summing over all $\alpha\in\ho$, and using the definition of the Brown
invariant.

If $K$ and $K'$ are equivalent Kasteleyn orientations, then $q^{K,\omega}_{D_0}=q^{K',\omega}_{D_0}$. In particular,
these two quadratic enhancements have the same Brown invariant. On the other hand, $\e^K(D_0)=(-1)^\mu\e^{K'}(D_0)$
and $\Pf(A^{K,\omega})=(-1)^\mu\Pf(A^{K',\omega})$, where $\mu$ is the number of vertices of
$\Gamma$ around which the orientation was flipped. Therefore, the term corresponding to $[K]$ in the statement of
the theorem does not depend on the choice of the representative in the equivalence class $[K]$.

Let us finally check that the coefficient
$i^{-\omega(D_0)}\exp(i\pi/4)^{\beta(q^{K,\omega}_{D_0})}\e^{K}(D_0)$ does not depend on $D_0$. Let $D$ be another dimer
configuration on $\Gamma$. By \cite[Lemma 3.7]{K-T} (where the sign needs to be corrected), and by the second part of
Theorem~\ref{thm:n-corr},
\[
\beta(q^{K,\omega}_{D_0})-\beta(q^{K,\omega}_{D})=2q^{K,\omega}_{D_0}([D+D_0]).
\]
On the other hand, we know by the beginning of the proof that
\[
\e^K(D_0)\e^K(D)\,i^{\omega(D)-\omega(D_0)}=i^{-q^{K,\omega}_{D_0}([D+D_0])}
\]
This concludes the proof of the theorem.
\end{proof}

Note that for any fixed $D_0$, one can always find an $\omega$ such that $\omega(D_0)=0$. Furthermore, for any
equivalence classe in $\K(X,\omega)$, one can choose its representative $K$ to satisfy $\e^K(D_0)=1$.
This leads to the formula stated in the introduction:
\[
Z=\frac{1}{2^{b_1/2}}\sum_{\eta\in\mathrm{Pin}^{-}(\SI)}\exp(i\pi/4)^{\beta(\eta)}\Pf(A^\eta),
\]
where $A^\eta$ is $A^{K,\omega}$ for any Kasteleyn orientation $K$ on $(X,\omega)$ such that $q^{K,\omega}_{D_0}=\eta$
and $\e^K(D_0)=1$.

\begin{remark}
Of course, the right-hand side of the equality in Theorem~\ref{thm:Pf-n}
does not depend on the choice of $\omega$ representing $w_1$,
as the left-hand side does not. Using the last part of Theorem~\ref{thm:n-corr}, one can make this statement a little more precise.
Given any two choices $\omega,\omega'$, let $\varphi_{\omega'\omega}\colon\K(X,\omega)\to\K(X,\omega')$ be the
equivariant bijection of Proposition~\ref{prop:omega}. Then, the summand corresponding to $[K]$ in the Pfaffian
formula given by $\omega$ is equal to the summand corresponding to $\varphi_{\omega'\omega}([K])$ in the Pfaffian
formula given by $\omega'$.
\end{remark}

\medskip

We finally come to the generalization of Theorem~\ref{thm:Pf'}, that is, the more hands-on version of the Pfaffian
formula. Recall that closed non-orientable surfaces fall into two categories.
\begin{romanlist}
\item{If $\chi(\SI)$ is odd, then $\SI$ is the connected sum of an orientable surface $\SI_g$ of genus $g\ge 0$ with
a projective plane $\R P^2$. A matrix of the modulo 2 intersection form is given by
$\begin{pmatrix}0&1\cr 1&0\end{pmatrix}^{\oplus g}\oplus(1)$.}
\item{If $\chi(\SI)$ is even, then $\SI$ is the connected sum of $\SI_g$ with
a Klein bottle $\K$. The modulo 2 intersection form admits the matrix
$\begin{pmatrix}0&1\cr 1&0\end{pmatrix}^{\oplus g}\oplus\begin{pmatrix}1&0\cr 0&1\end{pmatrix}$.}
\end{romanlist}

Let $\mathcal{B}=\{\alpha_j\}$ be a set of simple closed curves on $\SI$,
transverse to $\G$, whose classes form a basis of $H_1(\SI_g;\Z_2)\subset\ho$.
If $\SI$ has odd (resp. even) Euler characteristic, we also fix one simple closed curve $\beta_1$ (resp. two disjoint simple
closed curves $\beta_1,\beta_2$) on $\SI$, transverse to $\G$, disjoint from the $\alpha_j$'s, whose class
forms a basis of $H_1(\R P^2;\Z_2)$ (resp. $H_1(\K;\Z_2)$) in $\ho$. Define $\omega\in C^1(X;\Z_2)$ by
$\omega(e)=e\cdot\sum_\ell\beta_\ell$. It clearly represents the first Stiefel-Whitney class of $\SI$.

Fix a Kasteleyn orientation $K$ on $(X,\omega)$ so that $n^K(C_\gamma)$ is odd for each
$\gamma\in\{\alpha_j,\beta_\ell\}$, where $C_\gamma$ is an oriented closed curve in $\G$ associated to $\gamma$
as follows. Let $\Gamma'\subset\SI'$ denote the surface graph $\Gamma\subset\SI$ cut along $\sqcup_\ell\beta_\ell$, and endow
$\SI'$ with the counterclockwise orientation. For $\gamma=\alpha_j$, $C_\gamma$ is the oriented 1-cycle in $\G'\subset\SI'$
having $\alpha_j$ to its immediate left, meeting every vertex of $\Gamma'$ adjacent to $\alpha_j$ on this side.
(Moving $\beta_\ell$ if needed, it may be assumed that $C_\gamma$ is disjoint from $\sqcup_\ell\beta_\ell$ so that
$C_\gamma$ is a 1-cycle in $\G'$.) For $\gamma=\beta_\ell$, $C_\gamma$ is the oriented 1-cycle in $\G$ given by one edge
$e$ of $\G$ intersecting $\beta_\ell$ once, together with the oriented curve in $\G'$ joining the two endpoints of
$e$ in $\G'$ and having $\beta_\ell$ to its immediate left in $\SI'$. (If $\chi(\SI)$ is even, it may be assumed that
$C_{\beta_\ell}$ is disjoint from $\beta_{\ell'}$ for $\{\ell,\ell'\}=\{1,2\}$.)

For any $\eps=(\eps_1,\dots,\eps_{2g})\in\Z_2^{2g}$, let $K_{\eps}$ denote the Kasteleyn orientation
obtained from $K$ by inverting the orientation $K$ on the edge $e$ of $\G$ each time $e$ intersects $\alpha_j$
with $\eps_j=1$. Finally, if $\SI$ has even Euler characteristic,
let $K'_{\eps}$ be obtained by inverting $K_\eps$ on $e$ each time the edge $e$ intersects $\beta_1$.

\begin{theorem}
\label{thm:Pf'-n}
Let $\G$ be a graph embedded in a closed non-orientable surface $\SI$ such that $\SI\setminus\G$ consists of
open 2-discs. Then, the partition function of the dimer model on $\Gamma$ is given by
\[
Z=\frac{1}{2^{g}}\Big|\sum_{\eps\in\Z_2^{2g}}(-1)^{\sum_{j<k}\eps_j\eps_k\alpha_j\cdot\alpha_k}
\Big(\mathrm{Re}(\Pf(A^{K_\eps}))+\mathrm{Im}(\Pf(A^{K_\eps}))\Big)\Big|,
\]
if $\SI=\SI_g\#\R P^2$, and by
\[
Z=\frac{1}{2^{g}}\Big|\sum_{\eps\in\Z_2^{2g}}(-1)^{\sum_{j<k}\eps_j\eps_k\alpha_j\cdot\alpha_k}
\Big(\mathrm{Im}(\Pf(A^{K_\eps}))+\mathrm{Re}(\Pf(A^{K'_\eps}))\Big)\Big|,
\]
if $\SI=\SI_g\#\K$.
\end{theorem}
\begin{proof}
If $\G$ does not admit any dimer configuration, then all Pfaffians vanish by (\ref{equ:Pf}) and
our equalities hold. Therefore, it may be assumed that there exists a $D\in\D(\G)$. In particular,
$\G$ has an even number of vertices, so $\K(X,\omega)$ is an $\coho$-torsor by Theorem~\ref{thm:Kast-n}.
For $\eta=\eta_1\in\Z_2$ (resp. $\eta=(\eta_1,\eta_2)\in\Z^2_2$), let $K_{\eps,\eta}$ be obtained from
$K_\eps$ by inverting the orientation of an edge each time it intersects $\beta_1$ with $\eta_1=1$
(resp. and $\beta_2$ with $\eta_2=1$). The set $\{K_{\eps,\eta}\}_{(\eps,\eta)\in\Z_2^{b_1}}$ clearly
contains one element in each equivalence class of Kasteleyn orientations. Setting $\zeta=\exp(i\pi/4)$
and dropping the superscript $\omega$'s, Theorem~\ref{thm:Pf-n} implies
\begin{align*}
Z&=\frac{(-i)^{\omega(D)}}{2^{b_1/2}}
\sum_{(\eps,\eta)\in\Z_2^{b_1}}\zeta^{\beta(q^{K_{\eps,\eta}}_D)}\,\e^{K_{\eps,\eta}}(D)\,\Pf(A^{K_{\eps,\eta}})\\
&=\frac{1}{2^{b_1/2}}\,\Big|
\sum_{(\eps,\eta)\in\Z_2^{b_1}}\zeta^{\beta(q^{K_{\eps,\eta}}_D)-\beta(q^K_D)}\,\e^{K_{\eps,\eta}}(D)\e^K(D)
\,\Pf(A^{K_{\eps,\eta}})\Big|.\tag{$\star\star$}
\end{align*}
By \cite[Lemma 3.7]{K-T},
\[
\beta(q^{K_{\eps,\eta}}_D)-\beta(q^K_D)=-2q_D^K([\Delta_{\eps,\eta}]),
\]
where $[\Delta_{\eps,\eta}]\in\ho$ is determined by the fact that its Poincar\'e dual $[\Delta_{\eps,\eta}]^*$ satisfies
$q^K_D+2[\Delta_{\eps,\eta}]^*=q^{K_{\eps,\eta}}_D$. By Theorem~\ref{thm:n-corr}, this is equivalent to requiring that
$K+[\Delta_{\eps,\eta}]^*=K_{\eps,\eta}$. The definition of $K_{\eps,\eta}$ implies that
\[
\Delta_{\eps,\eta}=\sum_j\eps_j\alpha_j+\sum_\ell\eta_\ell\beta_\ell
\]
represents the right homology class. Using the equality
\[
\e^{K_{\eps,\eta}}(D)\e^{K}(D)=(-1)^{\Delta_{\eps,\eta}\cdot D}
\]
and the notation of Corollary~\ref{cor:n-corr}, we obtain that the coefficient in $(\star\star)$ corresponding to
$(\eps,\eta)$ is equal to
\[
\zeta^{-2q_D^K([\Delta_{\eps,\eta}])-4\Delta_{\eps,\eta}\cdot D}=(-i)^{q^K_\mathcal{B}([\Delta_{\eps,\eta}])}.
\]
Since $q^K_\mathcal{B}$ is a quadratic enhancement, and given the assumptions on the intersections of the cycles
of $\mathcal B$, it follows
\[
q^K_\mathcal{B}([\Delta_{\eps,\eta}])=
\sum_j\eps_j q^K_\mathcal{B}([\alpha_j])+2\sum_{j<k}\eps_j\eps_k\alpha_j\cdot\alpha_k+
\sum_\ell\eta_\ell q^K_\mathcal{B}([\beta_\ell]).
\]
We have chosen $K$ precisely so that $q^K_\mathcal{B}([\alpha_j])=0$ and
$q^K_\mathcal{B}([\beta_\ell])=-1$. Indeed, for $\gamma=\alpha_j$, $C_\gamma$ is constructed to satisfy
$\ell^\omega_D(C_\gamma)=D\cdot\gamma$ and $\omega(e)=0$ for any edge $e$ of $C_\gamma$. Therefore, the four terms of the sum
\[
q^K_\mathcal{B}([\gamma])=q^K_\mathcal{B}([C_\gamma])=2(n^K(C_\gamma)+1)+2(\ell^\omega_D(C_\gamma)+D\cdot\gamma)
+\omega(D\cap C_\gamma)-\omega(C_\gamma\setminus D)
\]
vanish. Similarly, for $\gamma=\beta_\ell$, $C_\gamma$ is constructed to satisfy
$\ell^\omega_D(C_\gamma)=D\cdot\gamma-\chi_D(e)$, where $\chi_D(e)=1$ is the edge $e$ of $C_\gamma$ crossing $\gamma$ is
occupied by a dimer of $D$, and $0$ otherwise. Since $\omega(D\cap C_\gamma)-\omega(C_\gamma\setminus D)=2\chi_D(e)-1$
and $n^K(C_\gamma)$ is odd, it follows that $q^K_\mathcal{B}([\gamma])=q^K_\mathcal{B}([C_\gamma])$ is equal to $-1$
as claimed.
Hence, we have the equality
\[
Z=\frac{1}{2^{b_1/2}}\Big|\sum_{(\eps,\eta)\in\Z_2^{b_1}}
(-1)^{\sum_{j<k}\eps_j\eps_k\alpha_j\cdot\alpha_k}\,i^{\sum_\ell\eta_\ell}\,\Pf(A^{K_{\eps,\eta}})\Big|.
\]
In the case of odd Euler characteristic, $A^{K_{\eps,1}}$ is nothing but the
complex conjugate of $A^{K_{\eps}}=A^{K_{\eps,0}}$. Therefore,
\begin{align*}
Z&=\frac{1}{2^{b_1/2}}\Big|\sum_{\eps\in\Z_2^{2g}}
(-1)^{\sum_{j<k}\eps_j\eps_k\alpha_j\cdot\alpha_k}\Big(\Pf(A^{K_\eps})+i\,\overline{\Pf(A^{K_\eps})}\Big)\Big|\\
	&=\frac{1}{2^{g}}\left|\frac{1+i}{\sqrt{2}}\right|\cdot
	\Big|\sum_{\eps\in\Z_2^{2g}}(-1)^{\sum_{j<k}\eps_j\eps_k\alpha_j\cdot\alpha_k}
	\Big(\mathrm{Re}(\Pf(A^{K_\eps}))+\mathrm{Im}(\Pf(A^{K_\eps}))\Big)\Big|.\\
\end{align*}
The case of even Euler characteristic is obtained similarly, using the fact that $A^{K_{\eps,1,1}}$
(resp. $A^{K_{\eps,0,1}}$) is the complex conjugate of $A^{K_{\eps,0,0}}=A^{K_\eps}$ (resp. $A^{K_{\eps,1,0}}=A^{K'_\eps}$).
\end{proof}

\begin{example}
\begin{figure}[htbp]
\labellist\small\hair 2.5pt
\pinlabel {$\beta_1$} at 330 430
\pinlabel {$\beta_1$} at 114 430
\pinlabel {$b$} at 465 220
\pinlabel {$b^{-1}$} at -30 220
\pinlabel {$\beta_2$} at 330 10
\pinlabel {$\beta_2$} at 114 10
\pinlabel {$C_{\beta_1}$} at -30 365
\pinlabel {$C_{\beta_2}$} at -30 80
\endlabellist
\centerline{\psfig{file=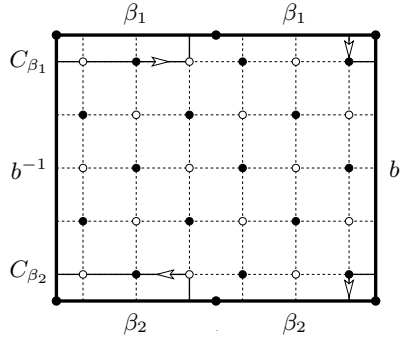,height=4cm}}
\caption{The cycles $C_{\beta_1}$ and $C_{\beta_2}$ on the bipartite graph $\G$.}
\label{fig:Klein2}
\end{figure}
Recall the graph $\G$ embedded in the Klein bottle as illustrated in Figure~\ref{fig:Klein}. With the notation
introduced above, the cycles $\beta_1$ and $\beta_2$ can be chosen to be the sides $a$ and $c$. Furthermore,
we can pick the oriented 1-cycles $C_{\beta_1}$ and $C_{\beta_2}$ as illustrated in Figure~\ref{fig:Klein2}.
If $K$ denotes the Kasteleyn orientation given in Figure~\ref{fig:Klein}, then both $n^K(C_{\beta_1})=1$ and
$n^K(C_{\beta_2})=3$ are odd as required. Therefore, Theorem~\ref{thm:Pf'-n} gives
\[
Z=\left|\mathrm{Im}(\Pf(A^{K}))+\mathrm{Re}(\Pf(A^{K'}))\right|,
\]
where $K'$ is obtained from $K$ by inverting the orientation of the three edges of $\G$ that cross $\beta_1$.
Note that $\G$ is a bipartite graph, as illustrated in Figure~\ref{fig:Klein2}. Therefore,
the Pfaffians can be computed by
\[
\Pf(A^K)=\Pf\begin{pmatrix}0&M\cr -M^T&0\end{pmatrix}=(-1)^{k(k-1)/2}\det(M),
\]
where $M$ is a square matrix of size $k$. In our case, we need to compute the determinant of two square matrices
of size $15$ (with $60$ non-zero coefficients each). Eventually, the number of dimer configurations on the
graph $\G$ is equal to
\[
Z=\left|\mathrm{Im}(\det(M))+\mathrm{Re}(\det(M'))\right|=20072.
\]
As a point of comparison, the $(5\times 6)$-square lattice embedded in the torus (that is, the graph $\Gamma$
above, but with the boundary identification $a^2bc^2b^{-1}$ replaced by
$acbc^{-1}a^{-1}b^{-1}$) has 9922 dimer configurations. Finally, the planar $(5\times 6)$-square lattice
(that is, the graph $\Gamma$ without the edges meeting the boundary of the hexagone) admits 1183 dimer configurations.
\end{example}

\subsection*{Acknowledgements}
We are grateful to Nicolai Reshetikhin for suggesting this problem, and for useful discussions.
We also thank Sebastian Baader, Mathieu Baillif and Giovanni Felder.
Finally, we thankfully acknowledge the hospitality of the University of Aarhus, the Louisiana State University
Baton Rouge, and the Indiana University Bloomington, where parts of this paper were done.

\end{document}